\documentclass[a4paper,UKenglish,cleveref, autoref, thm-restate]{lipics-v2021}
\nolinenumbers

\title{Optimality-Preserving Data Reduction for Maximum k-Cut (Full Version)} 

\author{Michael Kaibel\footnote{Corresponding author}}{University of Bonn, Germany}{mkaibel@uni-bonn.de}{https://orcid.org/0009-0006-1967-5376}{}

\author{Petra Mutzel}{University of Bonn, Germany \and Lamarr Institute, Bonn, Germany}{pmutzel@uni-bonn.de}{https://orcid.org/0000-0001-7621-971X}{}

\authorrunning{M. Kaibel and P. Mutzel} 

\Copyright{Michael Kaibel and Petra Mutzel} 



\category{} 

\supplement{Our source code is publicly available under \url{https://github.com/mkaibel/MaxKCutPreprocessing}}

\funding{This research was partially funded by the Deutsche Forschungsgemeinschaft (DFG, German Research Foundation) under grant FOR-5361 -- 459420781.}

\acknowledgements{The authors gratefully acknowledge the access to the Marvin cluster and the support provided by the High Performance Computing \& Analytics Lab of the University of Bonn}

\usepackage{todonotes}                                      
\usepackage{algorithm}                             
\usepackage[noend]{algpseudocode}                                  
\usepackage{subcaption}                                     
\usepackage{tikz}                                           
\usetikzlibrary{positioning, fit, calc}                     
\usepackage{scalerel}                                       
\usepackage{booktabs}                                       
\usepackage{mathabx}                                        
\usepackage{adjustbox}                                      

\newcommand{\cO}[0]{\mathcal{O}}
\newcommand{\cN}[0]{\mathcal{N}}

\newcommand{\cR}[0]{\mathcal{R}}

\newcommand{\cP}[0]{\mathcal{P}}

\newcommand{\cA}[0]{\mathcal{A}}

\newcommand{\cL}[0]{\mathcal{L}}

\newcommand{\bR}[0]{\mathbb{R}}

\newcommand\brac[1]{\left(#1\right)}
\newcommand\cbrace[1]{\left\{#1\right\}}

\newcommand{\mc}[0]{\textsc{MaxCut}}
\newcommand{\mclong}[0]{\textsc{Maximum Cut}}
\newcommand{\mkc}[1][]{\ifthenelse{\equal{#1}{}}{\textsc{Max} $k$-\textsc{Cut}}{\textsc{Max} $#1$-\textsc{Cut}}}
\newcommand{\mkclong}[1][]{\ifthenelse{\equal{#1}{}}{\textsc{Maximum} $k$-\textsc{Cut}}{\textsc{Max} $#1$-\textsc{Cut}}}

\newcommand{\np}[0]{\ensuremath{\cN \cP}}

\newcommand{\partition}[3]{#1^{\ifthenelse{\equal{#3}{}}{}{#3}}_{\ifthenelse{\equal{#2}{}}{}{|#2}}}

\newcommand{\tlc}[1]{G[#1]}
\newcommand{\tlcname}[0]{structured cut sets}
\newcommand{\tlctitle}[0]{\textsc{Structured Cut Sets}}
\newcommand{\tlcshort}[0]{\textsc{SCS}}
\newcommand{\bpc}[2]{\overline{#1[#2]}}

\begin{document}

\maketitle

\begin{abstract}
Preprocessing has become an increasingly important part of solving \mclong\@ to optimality, enabling exact solvers to tackle significantly larger instances. This suggests that exact solvers for the more general \mkclong\@ problem could also benefit from sophisticated preprocessing. However, to the best of our knowledge, no preprocessing techniques that are effective for $k > 2$ have been published.

In this paper, we introduce structured cut sets, a novel data reduction technique for \mkclong. We provide criteria under which deleting cut sets is optimality-preserving, yielding a decomposition into connected components that can be solved independently and whose solutions can be combined into an optimal solution for the original graph. 
Furthermore, we extend several preprocessing techniques from \mclong\@ to \mkclong. To show that our rules are optimality-preserving, we develop a new proof framework based on the addition of weighted graphs.

We complement our theoretical results by engineering a preprocessing framework for \mkclong\@ and show its effectiveness in a computational study. The preprocessed instances are typically significantly smaller. Integrating our preprocessing into an exact solver yields significant speed-ups and enables solving more instances to optimality.

\end{abstract}

\section{Introduction}
Graph partitioning problems are a fundamental class of optimization problems. A classic example is the \mkclong\@ problem (\mkc\@ for short), which asks for a partition of the nodes of a graph into up to $k$ subsets, maximizing the sum of weights of edges whose endpoints lie in different sets. An example is illustrated in Figure \ref{fig:k_cut_example}. The problem is \np-hard, which can be shown with a reduction from the related \textsc{$k$-Colouring} problem, one of Karp's 21 \np-complete problems \cite{Karp1972}.

\begin{figure}[h!]
    \captionsetup[subfigure]{justification=centering}
    \subcaptionbox{An instance}[0.3\linewidth]{
        \centering
        \begin{tikzpicture}[nodes={circle,draw}, scale=0.58]
        \node (1) at (0,0) { };
        \node (2) at (1,2) { };
        \node (3) at (2,0) { };
        \node (4) at (3,2) { };
        \node (5) at (4,0) { };
        \node (6) at (5,2) { };
        \node (7) at (6,0) { };

        \draw[-] (1) -- (2);
        \draw[-] (1) -- (3);
        \draw[-] (1) -- (4);
        \draw[-] (1) to[out=-45, in=-135] (7);
        \draw[-] (2) -- (3);
        \draw[-] (2) to[out=45, in=135] (6);
        \draw[-] (3) -- (4);
        \draw[-] (3) -- (5);
        \draw[-] (4) -- (6);
        \draw[-] (5) -- (6);
        \draw[-] (5) -- (7);
        \draw[-] (6) -- (7);
    \end{tikzpicture}
    \label{fig:sub:example_graph}
    }
    \hfill
    \subcaptionbox{A suboptimal solution}[0.3\linewidth]{
        \centering
        \begin{tikzpicture}[nodes={circle,draw}, scale=0.58]
        \node[fill=magenta] (1) at (0,0) { };
        \node[fill=blue] (2) at (1,2) { };
        \node[fill=blue] (3) at (2,0) { };
        \node[fill=orange] (4) at (3,2) { };
        \node[fill=magenta] (5) at (4,0) { };
        \node[fill=orange] (6) at (5,2) { };
        \node[fill=blue] (7) at (6,0) { };

        \draw[-] (1) -- (2);
        \draw[-] (1) -- (3);
        \draw[-] (1) -- (4);
        \draw[-] (1) to[out=-45, in=-135] (7);
        \draw[-, dotted] (2) -- (3);
        \draw[-] (2) to[out=45, in=135] (6);
        \draw[-] (3) -- (4);
        \draw[-] (3) -- (5);
        \draw[-, dotted] (4) -- (6);
        \draw[-] (5) -- (6);
        \draw[-] (5) -- (7);
        \draw[-] (6) -- (7);
    \end{tikzpicture}
    \label{fig:sub:example_suboptimal}
    }
    \hfill
    \subcaptionbox{An optimum solution}[0.3\linewidth]{
        \centering
        \begin{tikzpicture}[nodes={circle,draw}, scale=0.58]
        \node[fill=magenta] (1) at (0,0) { };
        \node[fill=blue] (2) at (1,2) { };
        \node[fill=orange] (3) at (2,0) { };
        \node[fill=blue] (4) at (3,2) { };
        \node[fill=magenta] (5) at (4,0) { };
        \node[fill=orange] (6) at (5,2) { };
        \node[fill=blue] (7) at (6,0) { };

        \draw[-] (1) -- (2);
        \draw[-] (1) -- (3);
        \draw[-] (1) -- (4);
        \draw[-] (1) to[out=-45, in=-135] (7);
        \draw[-] (2) -- (3);
        \draw[-] (2) to[out=45, in=135] (6);
        \draw[-] (3) -- (4);
        \draw[-] (3) -- (5);
        \draw[-] (4) -- (6);
        \draw[-] (5) -- (6);
        \draw[-] (5) -- (7);
        \draw[-] (6) -- (7);
    \end{tikzpicture}
    \label{fig:sub:example_optimal}
    }
    \caption{
        An example of \mkclong\@ on a unit weight graph with $k = 3$. All edges have weight $1$. Partitions are encoded by colour.
    }
    \label{fig:k_cut_example}
\end{figure}
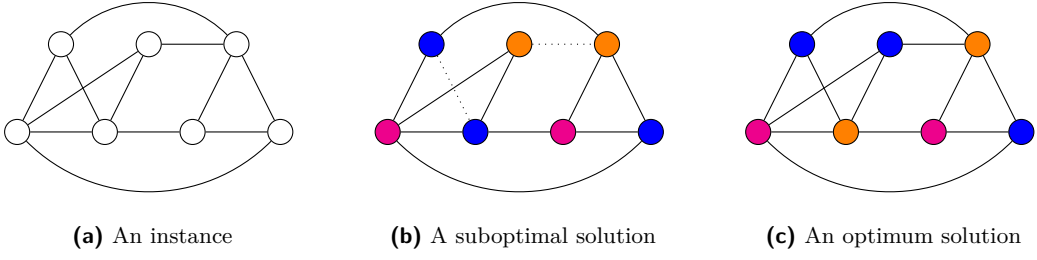

Research into \mkc\@ has been motivated by its various applications, e.g., in computing relaxations of frequency assignment problems \cite{eisenblatter2002frequency}, chip load balancing \cite{hendrickson1995improved, walshaw1997parallel}, and image reconstruction \cite{dahl2007integer}. Special attention has been placed on the case $k = 2$, called \mclong\@ (\mc\@ for short). Recent research has shown that preprocessing can significantly reduce the size of real world instances for \mc\@ by locating substructures for which the behaviour of the optimum solution can be determined ahead of time \cite{ferizovic2020engineering, rehfeldt2023faster, charfreitag2024separator, lange2019combinatorial}. This preprocessing enables exact solvers to tackle significantly larger instances, sometimes even solving instances to optimality purely during preprocessing. 
The only published work on preprocessing for \mkc\@ that we know of is \cite{fakhimi2025folding}. However, in their computational experiments their technique only had an effect for $k = 2$. The success of preprocessing for \mc\@ suggests that \mkc\@ could still benefit from sophisticated preprocessing. 

\subsection*{Our Contribution}
We extensively study optimality-preserving preprocessing techniques for the \mkc\@ problem. In detail: 
\begin{itemize}
    \item We introduce \tlcname, a new \mkc\@ data separation rule particularly effective for $k\ge 3$. The rule relies on novel criteria under which the problem reduces to independently solving the connected components obtained by deleting a cut set.
    \item Moreover, we generalize several of the \mc\@ preprocessing rules from \cite{charfreitag2024separator, ferizovic2020engineering, rehfeldt2023faster, lange2019combinatorial}.
    \item In order to prove that preprocessing rules for \mkc\@ are optimality-preserving, we introduce a new proof framework based on the addition of weighted graphs.
    \item We propose a preprocessing algorithm based on our theoretical findings, equipped with a framework for reconstructing optimal solutions.
    \item Our computational study shows that instances preprocessed by our algorithm are typically significantly smaller. When integrated into an exact solver, our preprocessing algorithm leads to substantial speed-ups and more instances solved to optimality.\footnote{The code is publicly available under \url{https://github.com/mkaibel/MaxKCutPreprocessing}}
 
\end{itemize}
\section{Preliminaries}
We consider weighted, undirected graphs $G = (V, E, w)$ with a set of vertices $V, |V| = n$, edges $E \subseteq \binom{V}{2}, |E| = m$ and weights $w: E \rightarrow \bR$. We denote by $\delta_G(v), N_G(v)$ and $d_G(v)$ the incident edges, neighbourhood and degree of $v \in V$ in $G$, by $w(S) = \sum_{e \in S} w(e)$ the weight of a set $S \subseteq E$ of edges and by $\delta_G(S) = \cbrace{\cbrace{s, t} \mid s \in S, t \not \in S}, N_G(S) = \bigcup_{v \in S} N(v) \setminus S$ the cut set and neighbourhood of $S \subseteq V$. If the graph is clear from context, we omit the subscript. 

We denote by $G[V'] = \brac{V', \cbrace{e \in E \mid e \subseteq V'}, w}$ the induced subgraph of a vertex set $V' \subseteq V$ and by $G[S] = \brac{\bigcup_{e \in S} e, S, w}$ the induced subgraph of an edge set $S \subseteq E$. For $V' \subseteq V$ we denote by $\partial_G(V') = \cbrace{v \in V' \mid \exists w \in V \setminus V': \cbrace{v, w} \in E}$ the boundary of $V'$.

\begin{definition}[$k$-Partition] \label{definition:k_partition}
    Given a set $M$ we call a function $p: M \rightarrow \{1,...,k\}$ a $k$-partition of $M$, which assigns every element $m \in M$ to group $p(m)$. \\
    For a subset $S \subseteq M$ we denote by $\partition{p}{S}{}: S \rightarrow \{1,...,k\}$, $\partition{p}{S}{}: x \mapsto p(x)$ the partition $p$ constrained to the set $S$. We denote by $S_{p = i} = \{s \in S \mid p(s) = i\}$ the set of all elements in $S$ assigned to group $i$.
\end{definition}

We refer to $\cbrace{1,...,k}$ as colours and to $p(v)$ as the colour assigned to $v$ by $p$.

\begin{definition}[Permutations of Partitions] \label{definition:permuting_partitions}
    Let $p: M \rightarrow \{1,...,k\}$ be a $k$-partition and $\pi: \{1,...,k\} \rightarrow \{1,...,k\}$ a bijection. We denote by $\partition{p}{}{\pi} = \pi \circ p$ the permutation of $p$ by $\pi$ and, for a subset $S \subseteq M$, by $\partition{p}{}{\pi(S)}$ the partition obtained by permuting $p$ with $\pi$ on $S$, that is
    \[\partition{p}{}{\pi(S)}(x) = \begin{cases}
        \pi(p(x)) & \text{if } x \in S \\
        p(x) & \text{otherwise}
    \end{cases}\]
    We say two $k$-partitions $p, p'$ are equivalent and write $p \cong p'$, iff $\partition{p}{}{\pi} = p'$ for a permutation $\pi$.
\end{definition}

Partial permutations of partitions will play a major role later, as permuting the partition of a subset $S \subseteq V$ of the nodes only changes how our $k$-cut behaves on $\delta(S)$.

Now that we have formally defined partitions we can define $k$-cuts and the considered problem in this work:

\begin{definition}[\mkclong] \label{definition:k_cut}
    Given a weighted undirected graph $G = (V, E, w)$ and a $k$-partition $p: V \rightarrow \{1,...,k\}$
    of the vertices we denote by
    \begin{align*}
        \delta_G(p) = \{\{v, w\} \in E \mid p(v) \ne p(w)\}
    \end{align*}
    the $k$-cut induced by $p$, that is the set of all edges such that their endpoints have different colours. We denote by $w(p) = w(\delta_G(p)) = \sum_{e \in \delta_G(p)} w(e)$ the weight of the $k$-cut induced by $p$. \mkc\@ is the problem of finding a $k$-partition $p$ that maximizes $w(p)$.
\end{definition}
Equivalent to \mkc\@ is the \textsc{Minimum $k$-Partition} problem, in which we minimize the weight of edges for which both endpoints have the same colour.

\subsection{Related Work}
As previously mentioned, on general graphs \mkc\@ is \np-hard. Even though \mkc\@ can likely not be solved in polynomial time, various approaches that can solve real world instances to optimality in reasonable time have been proposed over the years. 
These are usually based on ILP formalisms \cite{chopra1993partition, ales2016extended}. Different techniques to solve the ILP models have been proposed. The orbital fixing algorithm by Kaibel et. al \cite{kaibel2011orbitopal} efficiently prunes symmetries in the assignment formulation. To obtain better dual bounds, both cutting planes \cite{chopra1993partition, chopra1995facets}, and semidefinite relaxations \cite{ghaddar2011branch, anjos2013solving, van2016new, rodrigues2018computational, de2019improving, de2022computational} have been employed. For a recent study on exact solvers, we refer to Rodrigues de Sousa et.\@ al \cite{de2022computational}.

Several different approaches for preprocessing \mc\@ have been developed. Lange et al.\@ \cite{lange2019combinatorial} introduced criteria that show that certain edges with high weight are cut and certain edges with high negative weight are not cut in an optimum solution. They exploited this knowledge to then contract the endpoints. Rules that enable removing or simplifying substructures, such as unit weight cliques with a small boundary and induced $3$-paths, were developed by Ferizovic et al.\@ \cite{ferizovic2020engineering}. Recently, Charfreitag et al.\@ \cite{charfreitag2024separator} developed a framework than enables exploiting $2$ and $3$ separators in preprocessing. To our knowledge the only theoretical work on preprocessing for \mkc\@ is by Fakhimi et al.\@ \cite{fakhimi2025folding}. They introduce a folding techniques that allows contracting vertices under certain circumstances. However in their computational experiments their preprocessing only had an effect for $k = 2$ and could not reduce any instance for $k > 2$ beyond what trivial preprocessing could already accomplish.
\section{Data Reduction for \mkclong}
Data reduction for \mc\@ and \mkc\@ is based on two central components: Data transformations and data separations:

\begin{definition}[Data Transformation]
    A \textbf{data transformation} $\cA$ for \mkc\@ transforms a weighted graph $G = (V, E, w)$ into a weighted graph $G' = (V', E', w')$ and provides a reconstruction algorithm $\cR_{\cA}$ which maps any $k$-partition $p'$ of $V'$ to a $k$-partition $p$ of $V$. \\
    We call a \textbf{data transformation} a \textbf{data reduction} iff either $|V'| < |V|$ or $|V'| = |V|$ and $|E'| < |E|$.
\end{definition}

\begin{definition}[Data Separation]
    A \textbf{data separation} $\cA$ for \mkc\@ transforms a weighted graph $G = (V, E, w)$ into a number of weighted graphs $G'_1 = (V_1,E_1,w_1),...,G'_i = (V_i, E_i, w_i)$ and provides a reconstruction algorithm $\cR_{\cA}$, which maps $k$-partitions $p_1,...,p_i$ of $V_1,...,V_i$ to a $k$-partition $p$ of $V$.
\end{definition}

\begin{definition}[Optimality-Preserving]
    We call a data transformation (resp.\@ a separation) \textbf{optimality-preserving} if and only if $\cR_{\cA}$ maps any optimum \mkc\@ solution $p'$ (resp.\@ $p_1,...,p_i$) for $G'$(resp.\@ $G'_1,...,G'_i$) to an optimum solution $p$ for $G$.
\end{definition}
Our definitions build on the notion of data transformations introduced in \cite{charfreitag2024separator}. However, their definitions require that data transformations produce an offset $\beta$ between the optimum objective values of $G$ and $G'$. We decided to instead require a reconstruction algorithm, as we want to obtain not just the optimum solution value, but also a corresponding $k$-partition. Still, all data transformations/separations discussed in this paper also yield a constant offset. This enables computing lower and upper bounds for the original instance from lower and upper bounds of the reduced instance. For an overview of the offsets see Appendix \ref{appendix:offset_overview}.


Since our goal is to compute optimum solutions for \mkc, all data separations and transformations we consider are optimality-preserving. To reduce the verbosity we will omit specifying that a data separation/reduction rule is optimality-preserving outside of theorems.

\subsection{The Data Reduction Framework} \label{section:reduction_framework}
Previous research on \mc\@ preprocessing concerned itself with reducing the instances and obtaining the optimum solution value \cite{lange2019combinatorial, ferizovic2020engineering, rehfeldt2023faster,  charfreitag2024separator}. While the preprocessing rules discussed allow for reconstruction of an optimum solution for the original graph, the presented frameworks do not describe how to perform this reconstruction, nor is it clear how new preprocessing rules would be added to the reconstruction.

To this end, we introduce a new preprocessing framework that collects reconstruction algorithms during preprocessing. Our data reduction framework relies on repeatedly applying data transformation/separation rules until we can no longer reduce the graph. For reconstruction, we build a reconstruction rule tree: When applying a data reduction/separation rule to $G'$, we add its reconstruction algorithm as a child to the reconstruction algorithm of the rule that created $G'$. Once we cannot reduce anything any more, we solve the remaining graphs using some exact solver. We then apply the reconstruction rules in reverse DFS order, so whenever a reconstruction rule is called we have access to optimum solutions for all graphs that it depends on. A visualization of this can be seen in Figure \ref{fig:framework_example}.

\begin{figure}[!ht]
    \centering
    
    \begin{subfigure}[b]{0.15\textwidth}
        \hspace{\textwidth}
    \end{subfigure}
    \begin{adjustbox}{minipage=0.33\textwidth,frame}
    \begin{subfigure}[b]{\textwidth}
        \centering
        \begin{tikzpicture}[nodes={circle,draw}, scale=0.55]
        \node (root) at (0,0) {$\dots$};
        \node[opacity=0] (r2) at (-1, 2) {$\cR_B$};

        \draw[dashed] (2,-1) -- (2,4);

        \node[opacity=0] (G2) at (4.5,2.5) {$G_2$};
        \node (G) at (4,2) {\Huge $G$};
    \end{tikzpicture}
    \end{subfigure}
    \end{adjustbox}
    \begin{subfigure}[b]{0.15\textwidth}
            \[\overset{\substack{\text{Apply separation}\\ \text{rule } A\vspace{4pt}}}{\scaleto{\Longrightarrow}{10pt}}\]
            \vspace{-1.0cm}
    \end{subfigure}
    \begin{adjustbox}{minipage=0.33\textwidth,frame}
    \begin{subfigure}[b]{\textwidth}
        \centering
        \begin{tikzpicture}[nodes={circle,draw}, scale=0.55]
            \node (root) at (0,0) {$\dots$};
            \node (r1) at (0, 2.5) {$\cR_A$};
            \node[opacity=0] (r2) at (-1, 2) {$\cR_B$};

            \draw[->] (root) -- (r1);

            \draw[dashed] (2,-1) -- (2,4);

            \node (G1) at (3,1.5) {$G_1$};
            \node (G2) at (4.5,2.5) {$G_2$};
        \end{tikzpicture}
    \end{subfigure}
    \end{adjustbox}\\
    \vspace{0.2cm}
    \begin{subfigure}[b]{0.15\textwidth}
            \[\overset{\substack{\text{Apply reduction}\\ \text{rule } B \text{ to } G_1\vspace{4pt}}}{\scaleto{\Longrightarrow}{10pt}}\]
            \vspace{-0.5cm}
    \end{subfigure}
    \begin{adjustbox}{minipage=0.33\textwidth,frame}
    \begin{subfigure}[b]{\textwidth}
        \centering
        \begin{tikzpicture}[nodes={circle,draw}, scale=0.55]
            \node (root) at (0,0) {$\dots$};
            \node (r1) at (0, 2.5) {$\cR_A$};
            \node (r2) at (-1, 5) {$\cR_B$};

            \draw[->] (root) -- (r1);
            \draw[->] (r1) -- (r2);

            \draw[dashed] (2,-1) -- (2,6);

            \node (G1) at (3,1.5) {$G'_1$};
            \node (G2) at (4.5,2.5) {$G_2$};
        \end{tikzpicture}
    \end{subfigure}
    \end{adjustbox}
    \begin{subfigure}[b]{0.15\textwidth}
            \[\overset{\substack{\text{Apply reduction}\\ \text{rule } C \text{ to } G_2\vspace{4pt}}}{\scaleto{\Longrightarrow}{10pt}}\]
            \vspace{-0.5cm}
    \end{subfigure}
    \begin{adjustbox}{minipage=0.33\textwidth,frame}
    \begin{subfigure}[b]{\textwidth}
        \centering
        \begin{tikzpicture}[nodes={circle,draw}, scale=0.55]
            \node (root) at (0,0) {$\dots$};
            \node (r1) at (0, 2.5) {$\cR_A$};
            \node (r2) at (-1, 5) {$\cR_B$};
            \node (r3) at (1, 5) {$\cR_C$};

            \draw[->] (root) -- (r1);
            \draw[->] (r1) -- (r2);
            \draw[->] (r1) -- (r3);

            \draw[dashed] (2,-1) -- (2,6);

            \node (G1) at (3,1.5) {$G'_1$};
            \node (G2) at (4.5,2.5) {$G'_2$};
        \end{tikzpicture}
    \end{subfigure}
    \end{adjustbox}
    \caption{
        The preprocessing framework building the reconstruction tree (left) while applying data separation and reduction rules to the graphs (right).
    }
    \label{fig:framework_example}
\end{figure}
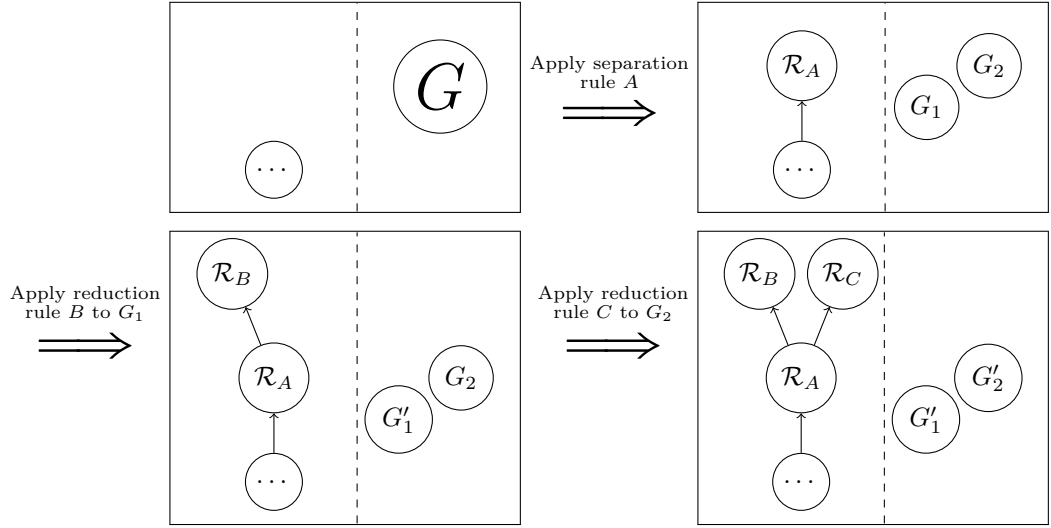

\subsection{Criteria for Preserving Optimality} \label{section:preprocessing_proof_framewoirk}
As many of the following proofs will use similar techniques, we establish some notation and results that provide a general framework for our later proofs to use.

\begin{definition}[Sum of Weighted Graphs]
    Given two weighted graphs $G_1 = (V_1, E_1, w_1)$ and $G_2 = (V_2, E_2, w_2)$, let $E_1 \cap E_2 = F$. We define their sum $G = (V, E, w) = G_1 + G_2$ with $V = V_1 \cup V_2$, $E = E_1 \cup E_2$ and
    \begin{align*}
        w(e) = \begin{cases}
            w_1(e) + w_2(e) & \text{if } e \in F \\
            w_1(e) & \text{if } e \in E_1 \setminus F \\
            w_2(e) & \text{if } e \in E_2 \setminus F
        \end{cases}
    \end{align*}
\end{definition}


We note that if we consider graphs that can be transformed into each other only by adding/deleting weight $0$ edges and degree $0$ vertices, this notion of $+$ forms a group. With this, we can now establish an optimality criterion via the addition of weighted  graphs.

\begin{theorem} \label{theorem:graph_sum_optimality}
    Let $G_1 = (V_1, E_1, w_1)$ and $G_2 = (V_2, E_2, w_2)$ be weighted graphs and $G_1 + G_2 = G = (V, E, w)$. Let $p$ be a $k$-partition of $V$ such that $\partition{p}{V_1}{}$ is an optimum solution for $G_1$ and $\partition{p}{V_2}{}$ is an optimum solution for $G_2$. Then $p$ is an optimum solution for $G$.
\end{theorem}

\begin{proof}
We note that for any $k$-partition $p'$ we have $w(p') = w_1(\partition{p'}{V_1}{}) + w_2(\partition{p'}{V_2}{})$. Using that $\partition{p}{V_1}{}, \partition{p}{V_2}{}$ are optimum solutions for $G_1$ and $G_2$, respectively, we get
$w(p) = w_1(\partition{p}{V_1}{}) + w_2(\partition{p}{V_2}{}) \geq w_1(\partition{p'}{V_1}{}) + w_2(\partition{p'}{V_2}{}) = w(p')$.
\end{proof}
We note that in general the optimum solution $p$ for $G$ is not optimum for $G_1$ or $G_2$.
However, given optimum solutions $p', p''$ for $G_1, G_2$ with $\partition{p'}{V_1 \cap V_2}{} \cong \partition{p''}{V_2 \cap V_2}{}$ we can efficiently combine these into an optimum solution for $G$:
\begin{lemma} \label{lemma:solution_reconstruction}
    Given two $k$-partitions $p'$ for $V_1$ and $p''$ for $V_2$ with $\partition{p'}{V_1 \cap V_2}{} \cong \partition{p''}{V_1 \cap V_2}{}$, we can compute a $k$-partition $p$ of $V_1 \cup V_2$ in $\cO(n)$ time such that $\partition{p}{V_1}{} \cong \partition{p'}{}{}$ and $\partition{p}{V_2}{} \cong \partition{p''}{}{}$.
\end{lemma}

\begin{proof}[Proof sketch]
    To compute $p$ from $p'$ and $p''$, we find a permutation $\pi$ with $\partition{p'}{V_1 \cap V_2}{} = \partition{p''}{V_1 \cap V_2}{\pi}$ and then permute $p''$ with $\pi$. For an algorithm and the full proof see Appendix \ref{appendix:criteria_preserving_optimality}.
\end{proof}
\section{Optimality-Preserving Data Reduction} \label{section:data_reduction}
In this section we extend several known data reduction rules from \mc\@ to \mkc.

Before we get to the more advanced rules, we briefly discuss a simple rule: We can remove any vertex $v \in V$ with $d(v) < k$ and only positive weight incident edges. In the reconstruction we colour $v$ with a colour not used in $N(v)$, of which there must be at least one.

\subsection{Cliques with a small Neighbourhood}
Cliques are interesting in the context of \mkc-preprocessing. For a clique $C$ with positive unit weight edges, any $k$-partition $p$ with $\forall i: \left\lceil \frac{|C|}{k}\right\rceil \geq |V_{p = i}| \geq \left\lfloor \frac{|C|}{k}\right\rfloor$ is an optimum solution. Let $G = (V, E, w)$ be a graph that is the sum of a unit weight clique $C = (V_C, E_C, w_C)$ and $G' = (V', E', w')$. If the intersection $V_C \cap V'$ is small enough, we can solve $G'$ to optimality and then colour the remaining vertices in $V_C$ such that our partition is also optimum for $C$. This idea has been employed for \mc\@ in \cite{ferizovic2020engineering} and was then generalised in \cite{charfreitag2024separator}. Their results extend to \mkc\@ very naturally.

\begin{theorem}[Same Neighbourhood Clique (generalises Proposition 4.5 in \cite{charfreitag2024separator})] \label{theorem:max_cut_cliques}
    Given a weighted undirected graph $G = (V, E, w)$ and a clique subgraph $C = (V_C, E_C, w)$ of $G$. If, for some constant $c > 0$, we have
    \begin{itemize}
        \item $\forall v \in V_C, e \in \delta_G(v): w(e) = c$, i.e.\@ $C$ is unit weight
        \item $\forall v \in V_C: N(v) \setminus V_C = N(V_C)$, i.e.\@ all $v \in V_C$ have the same neighbourhood outside $V_C$
        \item and $|N(V_C)|\leq \left\lceil \frac{|V_C \cup N(V_C)|}{k} \right\rceil$
    \end{itemize}
 then deleting all vertices in $V_C$ from $G$ and reducing the weights of all edges $\{v, w\}$ with $v, w \in N(C), v \ne w$ by $c$ (if two vertices are not connected we insert an edge with weight $-c$) is an optimality-preserving data reduction rule for \mkc.
\end{theorem}

The proof for \mkc\@ is analogous to the proof for \mc\@ in \cite{charfreitag2024separator}. While finding maximal cliques in general is difficult, cliques with a small neighbourhood can be located efficiently as noted in \cite{ferizovic2020engineering}.

\subsection{Triconnected Components and 2-Vertex-Separators}
\mc\@ can be solved efficiently on graphs with no $K_{3, 3}$-minor by exploiting $2$-vertex-separators, that is sets of two vertices such that their removal disconnects the graph \cite{chimani2019cut}. These ideas were adapted for \mc\@ preprocessing in \cite{charfreitag2024separator}. They generalise to \mkc:

\begin{theorem}[Two-Vertex-Separator reduction (Corollary 4.3 in \cite{charfreitag2024separator})] \label{theorem:two_vertex_cut}
    Given a weighted undirected graph $G = (V, E, w)$ and an induced subgraph $G[V']$ with $\partial_G(V') = \cbrace{u, v}$. Let $p'$ be the best \mkc\@ solution for $G[V']$ constrained by $u$ and $v$ having the same colour, and $p''$ the best solution in which they have different colours. Then deleting all vertices in $V' \setminus \cbrace{u, v}$ and setting $w(\{u, v\}) = w_1(p'') - w_1(p')$ is an optimality-preserving data reduction.
\end{theorem}

\begin{proof}[Proof sketch]
We briefly sketch the proof by \cite{charfreitag2024separator}, which also works for \mkc. Let $G'$ be the graph remaining after the reduction and $p$ be an ptimum solution for $G'$. If $u$ and $v$ have the same colour $p$, then we extend $p$ with $p'$ and otherwise we extend with $p''$. The change of $w(\cbrace{u, v})$ encodes the trade-off between these. For a full proof using our framework based on graph sums see the Appendix \ref{appendix:spqr_reduction}.
\end{proof}

To efficiently identify 2-vertex separators suitable for this reduction, we make use of the decomposition of G into its triconnected components \cite{hopcroft1973dividing}, and then attempt to remove the leaves in the resulting tree structure as suggested in \cite{charfreitag2024separator}. While \cite{charfreitag2024separator} removed such components only if they contained $\leq 21$ vertices, we noticed that branch-and-bound could frequently find the optimum solution for larger graphs as well. Due to this we instead opted to run an exact solver with a time limit of $1$ second.
\cite{charfreitag2024separator} employed similar techniques with $3$-vertex-separators. A brief discussion why this is not possible for \mkc\@ can be found in Appendix \ref{appendix:3_vertex_cuts}.

\subsection{Negative Weight Dominating Edges}
Another class of data reduction rules are dominating edges introduced by Lange et. al \cite{lange2019combinatorial}. 
The key idea is that any edge that is not cut in an optimum solution can be safely contracted without affecting optimality.

\begin{theorem}[Negative Dominating Edges (extends Theorem 1 (4) in \cite{lange2019combinatorial})] \label{theorem:dominating_mc}
    Let $G = (V, E, w)$ be a weighted graph, $\delta(S)$ be a cut set for some $S \subseteq V$ and $e \in \delta(S)$ with $w(e) < 0$. If
    \[-w(e) \geq \sum_{e' \in \delta(S) \setminus \{e\}} |w(e')|\]
    then contracting $e$ is an optimality-preserving data reduction.
\end{theorem}

\begin{proof}[Proof sketch]
    We adapt the proof by \cite{lange2019combinatorial}.  Let $e = \cbrace{u, v}$ such that $u \in S$. Let $p$ be an optimum solution. If $p$ cuts $e$, then we swap the colours $p(u), p(v)$ on $S$. In the worst case we stop cutting all positive weight edges and start cutting all negative weight edges in $\delta(S)$, except for $e$. However, due to $-w(e)$ being disproportionally high, this yields a $k$-partition $p'$ with $w(p') \geq w(p)$ that does not cut $e$. Thus, there is an optimum solution that does not cut $e$ and we can contract $u$ and $v$.
\end{proof}

Unfortunately, most other dominating edge rules in \cite{lange2019combinatorial, rehfeldt2023faster} and \cite{charfreitag2024separator} rely on specific properties of $k = 2$. A weaker version of the triangle rule by \cite{ferizovic2020engineering} and a discussion why the other rules do not generalise can be found in Appendix \ref{appendix:dominating_edges_extra}.
\section{Optimality-Preserving Data Separation} \label{section:data_separation}
We now move on to data separations. For these, we describe conditions under which we can solve several subgraphs independent of each other, instead of solving the whole graph.

\subsection{Separating (Bi)connected Components}
Separating graphs into their (bi)connected components is a well known data separation technique \cite{hochbaum1993should} and a standard preprocessing step for \mc\@ \cite{charfreitag2024separator, rehfeldt2023faster}. We briefly motivate why separating (bi)connected components is optimality-preserving for \mkc. As different (bi)connected components $C_1, C_2$ intersect in at most one vertex, optimum solutions $p', p''$ of $C_1, C_2$, respectively, have $\partition{p'}{C_1 \cap C_2}{} \cong \partition{p''}{C_1 \cap C_2}{}$. Therefore $p', p''$ can be combined into a $k$-partition $p$ of $V$ using Lemma \ref{lemma:solution_reconstruction}. By Theorem \ref{theorem:graph_sum_optimality} $p$ is optimum for $G$.

\subsection{\tlctitle} \label{section:two_layer_separators}
Here, we introduce a new separation rule called \emph{\tlcname\@}, which makes use of the combinatorial benefits of larger values of $k$. 




By definition, deleting the edges of a cut set $C \subseteq E$ from $G$ decomposes it into at least two disjoint components. In this section we assume w.l.o.g.\@ that deleting the cut set splits $G$ into exactly two connected components, $G_1 = \brac{V_1, E_1, w_1}$ and $G_2 = \brac{V_2, E_2, w_2}$. We also consider only cut sets in which all edges have positive edge weights (called positive cut sets). For a discussion why cut sets containing negative edges are disregarded, see Appendix \ref{appendix:further_applicability_two_layer}. 
Given optimum solutions $p', p''$ for $G_1, G_2$, we combine them into a $k$-partition $p$ of $V$ with $\forall v \in V_1: p(v) = p'(v)$ and $\forall v \in V_2: p(v) = p''(v)$. We now want to permute the colours on $V_2$ such that all edges in $C$ are cut by $p$. If this is possible, $p$ is optimum for $G$, as it is optimum for $G_1, G_2$ and $\tlc{C}$ and $G = G_1 + G_2 + \tlc{C}$. See Figure \ref{fig:tlc_swap_example} for an example.

\begin{figure}
    \centering
    \begin{subfigure}[b]{0.4\textwidth}
    \centering
    \begin{tikzpicture}[nodes={circle,draw}, scale=0.5]
        \node[fill=blue] (l1) at (0,0) {};
        \node[fill=orange] (l2) at (0,2) {};

        \draw[dashed] (l1) -- (-1, -0.7);
        \draw[dashed] (l1) -- (-1, 0);
        \draw[dashed] (l1) -- (-1, 0.7);
        \draw[dashed] (l2) -- (-1, 1.3);
        \draw[dashed] (l2) -- (-1, 2);
        \draw[dashed] (l2) -- (-1, 2.7);

        \draw [dashed,domain=-60:60] plot ({-2.5 + 3.5*cos(\x)}, {1+3.5*sin(\x)});

        \node[color=black!0, text=black!100] at (-1.7, 1) {\Large $G_1$};

        \node[fill=magenta] (r1) at (4,0) {};
        \node[fill=blue] (r2) at (4,2) {};

        \draw[dashed] (r1) -- (5, -0.7);
        \draw[dashed] (r1) -- (5, 0);
        \draw[dashed] (r1) -- (5, 0.7);
        \draw[dashed] (r2) -- (5, 1.3);
        \draw[dashed] (r2) -- (5, 2);
        \draw[dashed] (r2) -- (5, 2.7);

        \draw[-] (l1) -- (r1);
        \draw[, dotted] (l1) -- (r2);
        \draw[-] (l2) -- (r2);

        \draw [dashed,domain=120:240] plot ({6.5 + 3.5*cos(\x)}, {1+3.5*sin(\x)});

        \node[color=black!0, text=black!100] at (5.7, 1) {\Large $G_2$};
    \end{tikzpicture}
    \end{subfigure}
    \begin{subfigure}[b]{0.15\textwidth}
            \[\overset{\substack{\text{Permute colours}\\ \text{ on $V_2$}\vspace{4pt}}}{\scaleto{\Longrightarrow}{10pt}}\]
            \vspace{0.5cm}
    \end{subfigure}
    \begin{subfigure}[b]{0.4\textwidth}
    \centering
    \begin{tikzpicture}[nodes={circle,draw}, scale=0.5]
        \node[fill=blue] (l1) at (0,0) {};
        \node[fill=orange] (l2) at (0,2) {};

        \draw[dashed] (l1) -- (-1, -0.7);
        \draw[dashed] (l1) -- (-1, 0);
        \draw[dashed] (l1) -- (-1, 0.7);
        \draw[dashed] (l2) -- (-1, 1.3);
        \draw[dashed] (l2) -- (-1, 2);
        \draw[dashed] (l2) -- (-1, 2.7);

        \draw [dashed,domain=-60:60] plot ({-2.5 + 3.5*cos(\x)}, {1+3.5*sin(\x)});

        \node[color=black!0, text=black!100] at (-1.7, 1) {\Large $G_1$};

        \node[fill=orange] (r1) at (4,0) {};
        \node[fill=magenta] (r2) at (4,2) {};

        \draw[dashed] (r1) -- (5, -0.7);
        \draw[dashed] (r1) -- (5, 0);
        \draw[dashed] (r1) -- (5, 0.7);
        \draw[dashed] (r2) -- (5, 1.3);
        \draw[dashed] (r2) -- (5, 2);
        \draw[dashed] (r2) -- (5, 2.7);

        \draw[-] (l1) -- (r1);
        \draw[-] (l1) -- (r2);
        \draw[-] (l2) -- (r2);

        \draw [dashed,domain=120:240] plot ({6.5 + 3.5*cos(\x)}, {1+3.5*sin(\x)});

        \node[color=black!0, text=black!100] at (5.7, 1) {\Large $G_2$};
    \end{tikzpicture}
    \end{subfigure}
    \caption{We cut all edges in the cut set by permuting {\color{blue} blue}  $\rightarrow$ {\color{magenta} magenta} $\rightarrow$ {\color{orange} orange} $\rightarrow$ {\color{blue} blue} on $V_2$}
    \label{fig:tlc_swap_example}
\end{figure}
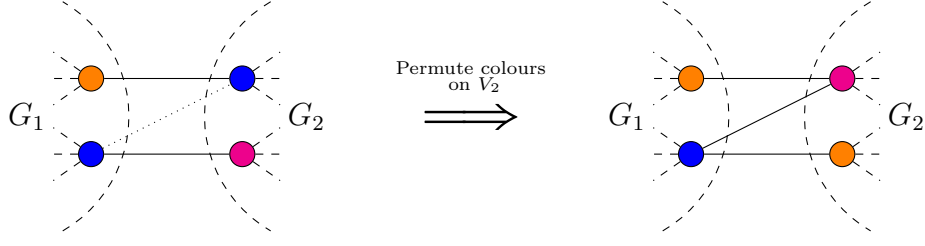

We develop criteria to decide whether such a colour permutation exists, independent of the optimal solutions for $G_1$ and $G_2$. This allows us to solve $G_1$ and $G_2$ independently, as we can be certain that all edges in $C$ can be cut. As our criteria work by showing that reconstruction is possible, we first introduce our reconstruction algorithm.

\subsubsection{Reconstruction for \tlctitle}
In order to compute an optimum solution $p$ of $G$ from partitions $p'$ and $p''$ of $G_1$ and $G_2$, we introduce the notion of a complement that preserves the bipartition of a graph:

\begin{definition}[Bipartite Complement Graph] \label{definition:bipartite_complement}
    Given an undirected bipartite graph $G = (L \sqcup R, E)$ such that $L, R$ is a bipartition of its vertices, we define the bipartite complement of $G$ with respect to $L, R$ as
    \begin{align*}
        \bpc{G}{L,R} = \brac{L \sqcup R, \cbrace{\cbrace{l, r} \mid l \in L, r \in R, \cbrace{l, r} \not \in E}}
    \end{align*}
\end{definition}

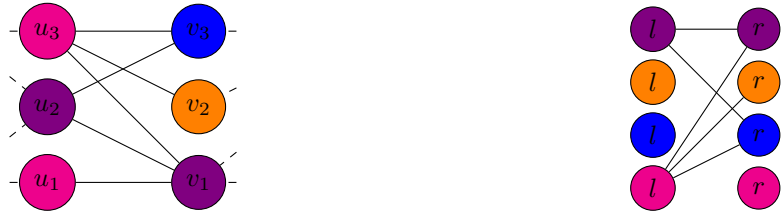
\begin{figure}[h!]
    \centering
    \begin{subfigure}[b]{0.5\textwidth}
        \centering
        \begin{tikzpicture}[nodes={circle,draw}, scale=0.5]
            \node[fill=magenta] (l1) at (0, 0) {$u_1$};
            \node[fill=violet] (l2) at (0, 2) {$u_2$};
            \node[fill=magenta] (l3) at (0, 4) {$u_3$};
            \node[fill=violet] (r1) at (4, 0) {$v_1$};
            \node[fill=orange] (r2) at (4, 2) {$v_2$};
            \node[fill=blue] (r3) at (4, 4) {$v_3$};

            \draw[-] (l1) -- (r1);
            \draw[-] (l2) -- (r1);
            \draw[-] (l2) -- (r3);
            \draw[-] (l3) -- (r1);
            \draw[-] (l3) -- (r2);
            \draw[-] (l3) -- (r3);

            \draw[dashed] (-1, 0) -- (l1);
            \draw[dashed] (-1, 1.2) -- (l2);
            \draw[dashed] (-1, 2.8) -- (l2);
            \draw[dashed] (-1, 4) -- (l3);

            \draw[dashed] (5, 0) -- (r1);
            \draw[dashed] (5, 0.8) -- (r1);
            \draw[dashed] (5, 2.5) -- (r2);
            \draw[dashed] (5, 4) -- (r3);
    \end{tikzpicture}
    \end{subfigure}
    \hfill
    \begin{subfigure}[b]{0.4\textwidth}
        \centering
        \begin{tikzpicture}[nodes={circle,draw}, scale=0.35]
            \node[fill=magenta] (l1) at (0, 0) {$l$};
            \node[fill=blue] (l2) at (0, 2) {$l$};
            \node[fill=orange] (l3) at (0, 4) {$l$};
            \node[fill=violet] (l4) at (0, 6) {$l$};
            \node[fill=magenta] (r1) at (4, 0) {$r$};
            \node[fill=blue] (r2) at (4, 2) {$r$};
            \node[fill=orange] (r3) at (4, 4) {$r$};
            \node[fill=violet] (r4) at (4, 6) {$r$};
            
            \draw[-] (l1) -- (r2);
            \draw[-] (l1) -- (r3);
            \draw[-] (l1) -- (r4);
            \draw[-] (l4) -- (r2);
            \draw[-] (l4) -- (r4);

    \end{tikzpicture}
    \end{subfigure}
    \caption{
        A cut set and its colour relation graph for $k = 4$ with the partition encoded via colours.
    }
    \label{fig:colour_relations_graph}
\end{figure}

The central idea of our algorithm is to construct a colour relation graph $G_{CR}$, which has $2k$ vertices. We have one vertex for every colour $i \in \cbrace{1,...,k}$ and each side $L, R$. Two vertices $(i, L), (j, R)$ are connected in $G_{CR}$ iff an edge $\{v, w\} \in C, v \in V_1, w \in V_2$ connects vertices with $p(v) = i, p(w) = j$. An example can be seen in Figure \ref{fig:colour_relations_graph}. These edges in $G_{CR}$ encode that we are not allowed to map $j$ to $i$ when permuting the colours on $V_2$. Correspondingly, edges in the bipartite complement of $G_{CR}$ encode how we are allowed to permute $p$ on $V_2$. Any perfect matching will give us a valid permutation $\pi$. If we permute $p$ with $\pi$ on $V_2$, all edges in $C$ will be cut. Correspondingly, given a solution $p$ for $G$ with $\partition{p}{V_1}{} \cong p', \partition{p}{V_2}{} \cong p''$ that cuts all edges in $C$, we can derive a perfect matching in the colour relation graph. We do this by determining a permutation $\pi$ that maps $p''$ into $\partition{p}{V_2}{}$.

From this idea we can derive the following reconstruction algorithm:

\begin{algorithm}[H] 
    \caption{\textsc{Reconstruct\tlcshort}$(G = (V, E, w), C, p', p'', k)$}
    \begin{algorithmic}[1]
        \State Find connected components $G_1 = (V_1, E_1, w_1), G_2 = (V_2, E_2, w_2)$ in $G$ without $C$
        \State Initialise $G_{CR} = (\{1,...,k\} \times \{L, R\}, \{\{(p'(v),L), (p''(w),R)\} \mid \{v, w\} \in C, v \in V_1, w \in V_2\})$
        \State Compute $M = \textsc{MaximumMatching}(\bpc{G_{CR}}{\{(i, L) \mid i \in \{1,...,k\}\}, \{(i,R) \mid i \in \{1,...,k\}\}})$
        \If{$|M| < k$}
            \State \textbf{return} \texttt{RECONSTRUCTION NOT POSSIBLE}
        \EndIf
        \State Initialise array $\pi$ of length $k$ such that $\forall \cbrace{(i, L), (j, R)} \in M: \pi[j] = i$
        \State Compute $k$-partition $p$ of $V$ with $\forall v \in V_1: p(v) = p'(v)$ and $\forall v \in V_2: p(v) = \pi[p''(v)]$
        \State \textbf{return} $p$
    \end{algorithmic}
\end{algorithm}

\begin{theorem} \label{theorem:two_layer_reconstructor}
    Given $G, C, p'$ and $p''$, \textsc{Reconstruct\tlcshort} finds a $k$-partition $p$ of $V$ that cuts all $e \in C$ with $\partition{p}{V_1}{} \cong p'$ and $\partition{p}{V_2}{} \cong p''$ if one exists, and returns an error otherwise. Its runtime is in $\cO(n + |C| + k^{2.5})$.
\end{theorem}

The proof follows the previous outline and can be found in full in Appendix \ref{appendix:two_layer_reconstructor}. By applying Theorem \ref{theorem:graph_sum_optimality} we obtain the following result:

\begin{corollary} \label{corollary:two_layer_reconstructor}
    Let $C$ be a positive cut set that partitions $G$ into $G_1$ and $G_2$ with optimum solutions $p'$ and $p''$, respectively. If \textsc{Reconstruct\tlcshort} returns a partition $p$, then $p$ is optimum for $G$.
\end{corollary}

We now know how to efficiently reconstruct a cut set, provided that reconstruction is possible. However, as seen in Figure \ref{fig:two_layer_doesnt_work}, there are cut sets for which reconstruction is not guaranteed to succeed.

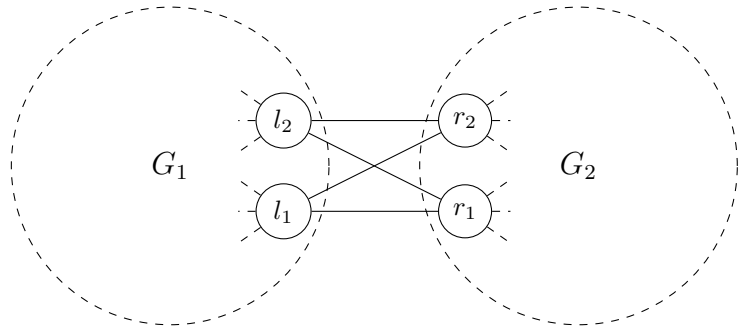
\begin{figure}[h!]
    \centering
    \begin{tikzpicture}[nodes={circle,draw}, scale=0.6]
        \node (l1) at (0,0) {$l_1$};
        \node (l2) at (0,2) {$l_2$};

        \draw[dashed] (l1) -- (-1, -0.7);
        \draw[dashed] (l1) -- (-1, 0);
        \draw[dashed] (l1) -- (-1, 0.7);
        \draw[dashed] (l2) -- (-1, 1.3);
        \draw[dashed] (l2) -- (-1, 2);
        \draw[dashed] (l2) -- (-1, 2.7);

        \draw[dashed] (-2.5, 1) circle (3.5);

        \node[color=black!0, text=black!100] at (-2.5, 1) {\Large $G_1$};

        \node (r1) at (4,0) {$r_1$};
        \node (r2) at (4,2) {$r_2$};

        \draw[dashed] (r1) -- (5, -0.7);
        \draw[dashed] (r1) -- (5, 0);
        \draw[dashed] (r1) -- (5, 0.7);
        \draw[dashed] (r2) -- (5, 1.3);
        \draw[dashed] (r2) -- (5, 2);
        \draw[dashed] (r2) -- (5, 2.7);

        \draw[-] (l1) -- (r1);
        \draw[-] (l1) -- (r2);
        \draw[-] (l2) -- (r1);
        \draw[-] (l2) -- (r2);

        \draw[dashed] (6.5, 1) circle (3.5);

        \node[color=black!0, text=black!100] at (6.5, 1) {\Large $G_2$};

    \end{tikzpicture}
    \caption{Let $k = 3$, $p'$ requires that $l_1$ and $l_2$ have different colours and $p''$ requires that $r_1$ and $r_2$ have different colours. Then reconstructing the cut set is impossible.}
    \label{fig:two_layer_doesnt_work}
\end{figure}

Therefore, before splitting $G$ at a positive cut set, we must make sure that \textsc{Reconstruct\tlcshort} finds a partition. We present some efficiently checkable criteria that assure that we can reconstruct a solution or prove that this is not possible. They make no assumptions about $p'$ and $p''$, and therefore work regardless of the partitions provided for $V_1$ and $V_2$.

\subsubsection{Criteria for Removing \tlctitle}
We first introduce two criteria showing that sufficiently small cut sets can always be removed.

\begin{theorem} \label{theorem:two_layer_sufficient}
    Let $G = (V, E, w)$ be a weighted graph and $C$ be a positive cut set. If
    \begin{itemize}
        \item $\tlc{C}$ has at most $k$ vertices or
        \item $\tlc{C}$ has at most $k-1$ edges
    \end{itemize}
    then splitting $G$ into $G_1, G_2$ by deleting $C$ is an optimality-preserving data separation.
\end{theorem}

\begin{proof}[Proof sketch]
     If there are at most $k$ vertices in $\tlc{C}$, then we can permute the colours on $G_2$ such that no vertices on the left and right side share a colour. If there are less than $k$ edges in the cut set and we permute the colours on $G_2$ uniformly at random, then in expectation there are $|C| \cdot \frac{1}{k} < 1$ edges in $C$ we do not cut. Therefore there must be a permutation that cuts all edges in $C$. For an algorithmic proof see Appendix \ref{appendix:two_layer_sufficient}.
\end{proof}

Both of these criteria can be checked efficiently. However, as we can see in Figure \ref{fig:two_layer_weak_criterion_counterexample}, there are cut sets that we can split that satisfy neither criterion.

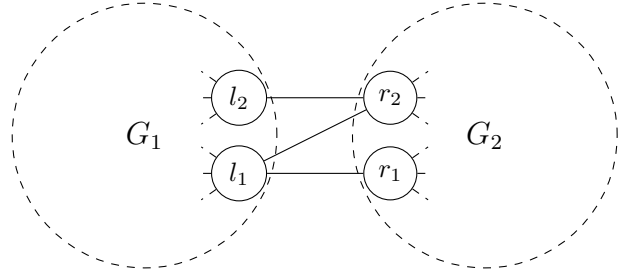
\begin{figure}[h!]
    \centering
    \begin{tikzpicture}[nodes={circle,draw}, scale=0.5]
        \node[] (l1) at (0,0) {$l_1$};
        \node[] (l2) at (0,2) {$l_2$};

        \draw[dashed] (l1) -- (-1, -0.7);
        \draw[dashed] (l1) -- (-1, 0);
        \draw[dashed] (l1) -- (-1, 0.7);
        \draw[dashed] (l2) -- (-1, 1.3);
        \draw[dashed] (l2) -- (-1, 2);
        \draw[dashed] (l2) -- (-1, 2.7);

        \draw[dashed] (-2.5, 1) circle (3.5);

        \node[color=black!0, text=black!100] at (-2.5, 1) {\Large $G_1$};

        \node[] (r1) at (4,0) {$r_1$};
        \node[] (r2) at (4,2) {$r_2$};

        \draw[dashed] (r1) -- (5, -0.7);
        \draw[dashed] (r1) -- (5, 0);
        \draw[dashed] (r1) -- (5, 0.7);
        \draw[dashed] (r2) -- (5, 1.3);
        \draw[dashed] (r2) -- (5, 2);
        \draw[dashed] (r2) -- (5, 2.7);

        \draw[-] (l1) -- (r1);
        \draw[-] (l1) -- (r2);
        \draw[-] (l2) -- (r2);

        \draw[dashed] (6.5, 1) circle (3.5);

        \node[color=black!0, text=black!100] at (6.5, 1) {\Large $G_2$};

    \end{tikzpicture}
    \caption{A cut set. For $k = 3$, neither criterion from Theorem \ref{theorem:two_layer_sufficient} holds. However, as $l_2$ and $r_1$ can have the same colour, we can always reconstruct.}
    \label{fig:two_layer_weak_criterion_counterexample}
\end{figure}

A naive approach to solve this issue would be to simply enumerate all colourings of $\partial_G(V_1), \partial_G(V_2)$ and check if we can reconstruct. However, this would be very time consuming, even for small $k$. Instead we present another criterion that can be applied to larger separators and takes their combinatorial structure into account:

\begin{theorem} \label{theorem:two_layer_sufficient_matching}
    Let $G = (V, E, w)$ be a weighted graph and $C$ be a positive cut set that splits $G$ into $G_1 = (V_1, E_1, w_1)$ and $G_2 = (V_2, E_2, w_2)$. We denote $L = \partial_G(V_1)$ and $R = \partial_G(V_2)$. \\
    If $\bpc{\tlc{C}}{L, R}$ contains a matching $M$ of size $\geq 2 \cdot (|V_C| - k) - 1$, then splitting $G$ into $G_1$ and $G_2$ by deleting $C$ is an optimality-preserving data separation
\end{theorem}

\begin{proof}[Proof sketch]
    Given $p'$ and $p''$, we can derive the colour relation graph from $\tlc{C}$ by repeatedly contracting vertices that have the same colour and side. We then add vertices for the colours that are missing. We call these newly added vertices jokers, as in the complement graph we can match them to any vertex on the other side. An example of this can be seen in Figure \ref{fig:tlc_matching_criterion}. If, for an $e \in M$ one of its endpoints is contracted, then we remove it. Therefore, for each contraction, we lose at most $2$ edges in $M$ while gaining a joker. The bound is chosen such that we can match all vertices not covered by the remaining edges in $M$ with jokers after any contraction sequence. For the full proof see Appendix \ref{appendix:two_layer_sufficient_matching}.
\end{proof}

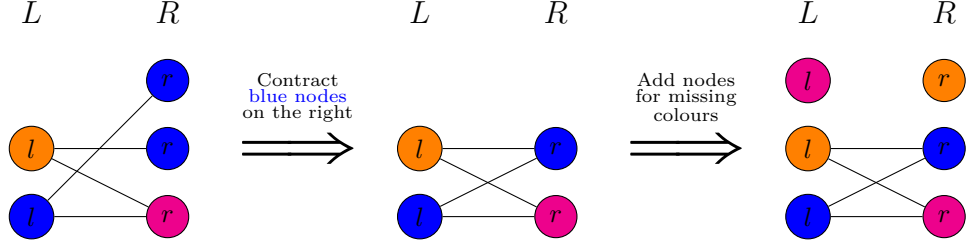
\begin{figure}[h!]
    \centering
    \begin{subfigure}[b]{0.23\textwidth}
    \centering
    \begin{tikzpicture}[nodes={circle,draw}, scale=0.45]
        \node[fill=blue] (l1) at (0,0) {$l$};
        \node[fill=orange] (l2) at (0,2) {$l$};

        \node[fill=magenta] (r1) at (4,0) {$r$};
        \node[fill=blue] (r2) at (4,2) {$r$};
        \node[fill=blue] (r3) at (4,4) {$r$};

        \draw[-] (l1) -- (r1);
        \draw[-] (l1) -- (r3);
        \draw[-] (l2) -- (r1);
        \draw[-] (l2) -- (r2);
        
        \node[color=black!0, text=black!100] at (0, 6) {\Large $L$};
        \node[color=black!0, text=black!100] at (4, 6) {\Large $R$};
    \end{tikzpicture}
    \end{subfigure}
    \begin{subfigure}[b]{0.12\textwidth}
            \[\overset{\substack{\text{Contract}\\ \text{{\color{blue} blue nodes}} \\
            \text{on the right}\vspace{4pt}}}{\scaleto{\Longrightarrow}{12pt}}\]
            \vspace{0.15cm}
    \end{subfigure}
    \begin{subfigure}[b]{0.23\textwidth}
    \centering
    \begin{tikzpicture}[nodes={circle,draw}, scale=0.45]
        \node[fill=blue] (l1) at (0,0) {$l$};
        \node[fill=orange] (l2) at (0,2) {$l$};

        \node[fill=magenta] (r1) at (4,0) {$r$};
        \node[fill=blue] (r2) at (4,2) {$r$};
        \node[color=black!0, text=black!0] (r3) at (4, 4) {$r$};

        \draw[-] (l1) -- (r1);
        \draw[-] (l1) -- (r2);
        \draw[-] (l2) -- (r1);
        \draw[-] (l2) -- (r2);

        \node[color=black!0, text=black!100] at (0, 6) {\Large $L$};
        \node[color=black!0, text=black!100] at (4, 6) {\Large $R$};
    \end{tikzpicture}
    \end{subfigure}
    \begin{subfigure}[b]{0.12\textwidth}
            \[\overset{\substack{\text{Add nodes}\\ \text{for missing} \\ \text{colours}\vspace{4pt}}}{\scaleto{\Longrightarrow}{12pt}}\]
            \vspace{0.15cm}
    \end{subfigure}
    \begin{subfigure}[b]{0.23\textwidth}
    \centering
    \begin{tikzpicture}[nodes={circle,draw}, scale=0.45]
        \node[fill=blue] (l1) at (0,0) {$l$};
        \node[fill=orange] (l2) at (0,2) {$l$};
        \node[fill=magenta] (l3) at (0, 4) {$l$};

        \node[fill=magenta] (r1) at (4,0) {$r$};
        \node[fill=blue] (r2) at (4,2) {$r$};
        \node[fill=orange] (r3) at (4, 4) {$r$};

        \draw[-] (l1) -- (r1);
        \draw[-] (l1) -- (r2);
        \draw[-] (l2) -- (r1);
        \draw[-] (l2) -- (r2);

        \node[color=black!0, text=black!100] at (0, 6) {\Large $L$};
        \node[color=black!0, text=black!100] at (4, 6) {\Large $R$};
    \end{tikzpicture}
    \end{subfigure}
    \caption{Obtaining $G_{CR}$ by contracting nodes sharing side and colour, then adding nodes for the missing colours}
    \label{fig:tlc_matching_criterion}
\end{figure}

This criterion can be checked efficiently in time $\cO(|C| + k^{2.5})$ using the algorithm by \cite{hopcroft1973n}. Further, if the cut set contains at most $k$ vertices, we only require $|M| \geq 0$. Therefore, this criterion implies the first criterion from Theorem \ref{theorem:two_layer_sufficient}.
Note that this criterion is sufficient, but for $k > 3$ it is not necessary. For $k = 4$ an example of a graph for which the criterion does not hold, but that is valid to separate, can be seen in Figure \ref{fig:two_layer_criterion_counterexample}.

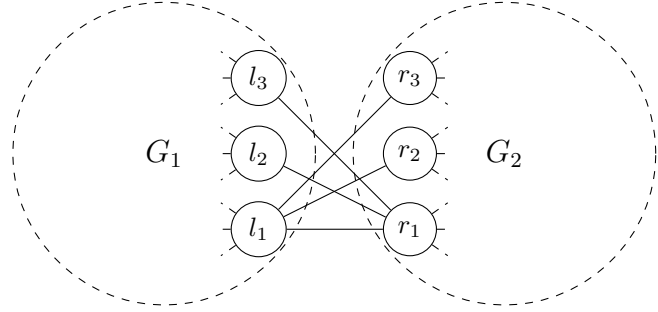
\begin{figure}[h!]
    \centering
    \begin{tikzpicture}[nodes={circle,draw}, scale=0.5]
        \node (l1) at (0,0) {$l_1$};
        \node (l2) at (0,2) {$l_2$};
        \node (l3) at (0,4) {$l_3$};

        \draw[dashed] (l1) -- (-1, -0.7);
        \draw[dashed] (l1) -- (-1, 0);
        \draw[dashed] (l1) -- (-1, 0.7);
        \draw[dashed] (l2) -- (-1, 1.3);
        \draw[dashed] (l2) -- (-1, 2);
        \draw[dashed] (l2) -- (-1, 2.7);
        \draw[dashed] (l3) -- (-1, 3.3);
        \draw[dashed] (l3) -- (-1, 4);
        \draw[dashed] (l3) -- (-1, 4.7);

        \draw[dashed] (-2.5, 2) circle (4);

        \node[color=black!0, text=black!100] at (-2.5, 2) {\Large $G_1$};

        \node (r1) at (4,0) {$r_1$};
        \node (r2) at (4,2) {$r_2$};
        \node (r3) at (4,4) {$r_3$};

        \draw[dashed] (r1) -- (5, -0.7);
        \draw[dashed] (r1) -- (5, 0);
        \draw[dashed] (r1) -- (5, 0.7);
        \draw[dashed] (r2) -- (5, 1.3);
        \draw[dashed] (r2) -- (5, 2);
        \draw[dashed] (r2) -- (5, 2.7);
        \draw[dashed] (r3) -- (5, 3.3);
        \draw[dashed] (r3) -- (5, 4);
        \draw[dashed] (r3) -- (5, 4.7);

        \draw[-] (l1) -- (r1);
        \draw[-] (l1) -- (r2);
        \draw[-] (l1) -- (r3);
        \draw[-] (l2) -- (r1);
        \draw[-] (l3) -- (r1);

        \draw[dashed] (6.5, 2) circle (4);

        \node[color=black!0, text=black!100] at (6.5, 2) {\Large $G_2$};

    \end{tikzpicture}
    \caption{If $k = 4$ we can always reconstruct here by assigning $l_1$ and $r_1$ to different groups and then assigning the remaining $2$ groups to $l_2, l_3, r_2, r_3$ if necessary. As all edges contain either $l_1$ or $r_1$ it is no problem if e.g. $l_2, r_3$ are assigned to the same group. Note that this \tlcname\@ fulfils none of our criteria.}
    \label{fig:two_layer_criterion_counterexample}
\end{figure}

To complement the criteria showing that reconstruction is possible, we also show a simple criterion that shows reconstruction is not guaranteed to be possible:
\begin{theorem} \label{theorem:no_big_two_layer_separators}
    Given a weighted undirected graph $G = (V, E, w)$ and a positive cut set $C$ in $G$ that splits $G$ into $G_1 = (V_1, E_1, w_1)$ and $G_2 = (V_2, E_2, w_2)$. If $\max\cbrace{|\partial_G(V_1)|, |\partial_G(V_2)|} \geq k$, then there exist $k$-partitions $p', p''$ of $V_1$ and $V_2$, such that there is no $k$-partition $p$ of $V$ with $\partition{p}{V_1}{} \cong p'$ and $\partition{p}{V_2}{} \cong p''$ $p$ that cuts all edges in $C$.
\end{theorem}

\begin{proof}
    Let $|\partial_G(V_1)| \geq k$. If $p'$ uses all $k$ colours vertices in $\partial_G(V_1)$ and $p''$ colours all vertices in $\partial_G(V_2)$ same colour, then at least one edge in $C$ can not be cut by permuting $p$ on $V_2$.
\end{proof}

This implies that for $k = 2$ we can only remove cut sets containing one edge, which are already covered by splitting biconnected components. 

\subsection{Criteria where one Solution is Known}
So far we assumed that the solutions for $G_1$ and $G_2$ are unknown. If we know the optimum solution for w.l.o.g.\@ $G_2$, we can contract all vertices in $V_2$ that share a colour. If any criterion from Theorems \ref{theorem:two_layer_sufficient} or \ref{theorem:two_layer_sufficient_matching} is fulfilled, we can delete all vertices in $V_2$. Knowing the solution for $G_2$ also enables a new criterion: Let $v_1,...,v_r$ be the boundary of the contracted $G_2$. If $d_{\tlc{C}}(v_i) \leq k - i$, we can delete all vertices in $V_2$. During reconstruction we permute $p$ on $V_2$ such that $v_i$ gets a colour not used by any vertex in $N_{\tlc{C}}(v_i)$ or $v_1,...,v_{i - 1}$. 

Given a positive cut set $C$ we can not split with Theorems \ref{theorem:two_layer_sufficient} or \ref{theorem:two_layer_sufficient_matching}. If $|V_2| \leq 20$ we solve $G_2$ to optimality and check the criteria above.

\subsection{Finding \tlctitle}
As we do not have a nice ``if and only if'' characterisation of which cut sets we can split, designing an efficient algorithm to find them is difficult. Theorem \ref{theorem:no_big_two_layer_separators} implies that we can only split \tlcname\@ with at most $(k-1)^2$ many edges. We therefore propose to find structured cut sets by identifying small cut sets, checking our criteria, and solving one side to optimality if it is sufficiently small. For this we employ a heuristic based on the randomized \textsc{MinimumCut} algorithm \textsc{FastCut} by Karger and Stein~\cite{karger1996new}. \textsc{FastCut} works by randomly contracting edges to generate large numbers of small cuts, which we take as candidates.
\section{Computational Experiments}
For our experiments, we implemented the assignment model for \mkc\@ by \cite{chopra1993partition}. To strengthen the model we employ greedy heuristics to separate cycle, general clique, wheel and bicycle wheel inequalities from \cite{chopra1993partition} and $1, -1$ hypermetric inequalities from \cite{chopra1995facets}. The solver serves both as a baseline to compare against and as the solver for the kernel graphs left over by our preprocessor. To provide a warm start for the solver we implemented a local search algorithm with tabu search based on \cite{ma2017multiple}. We also used this heuristic for experiments on very large instances.

For the order of data separation and reduction rules we chose to first employ rules that maintain local positivity and unit weight of edges, to ensure that rules requiring this can still be applied later. This lead us to the following order:
\begin{enumerate}
    \item Remove low degree vertices and split (bi)connected components
    \item Locate and remove \tlcname\@ using \textsc{FastCut} and Theorems \ref{theorem:two_layer_sufficient} and \ref{theorem:two_layer_sufficient_matching}.
    \item Remove unit weight cliques with a small neighbourhood (Theorem \ref{theorem:max_cut_cliques}).
    \item Contract negative-weight dominating edges.
    \item Find two vertex separators yielding at least one small component and apply Theorem \ref{theorem:two_vertex_cut}.
    \item Find cut sets where one side is small, solve it and check if it can be removed.
\end{enumerate}

If any rule successfully reduces the graph, we restart at the top. We use queues and flags to keep track of which vertices are promising candidates for the various rules and to make sure we do not run expensive rules too often. Similar techniques are used in \cite{ferizovic2020engineering, charfreitag2024separator}. We decided to forego the preprocessing rule by \cite{fakhimi2025folding}, as they reported no reduction for $k > 2$.

\subsection{Setup}
We implemented everything in \texttt{C++ 20}. For a graph class and fundamental algorithms we relied on NetworKit \cite{staudt2016networkit}. To locate biconnectors we use the SPQR-tree implementation of \cite{gutwenger2000linear} in the Open Graph Drawing Framework \cite{chimani2013open}. We compiled the code using gcc 14.3.0 with the ``-O3'' flag. All experiments were run on a  machine with 2 Intel Xeon ``Sapphire Rapids'' 48-core CPUs and 1024 GB RAM running AlmaLinux 9.6. To solve ILPs we used the state of the art solver Gurobi 13 \cite{gurobi} and set it to use as many threads as needed. We also employed parallelism to collect small cuts for separating \tlcname\ and in our local search heuristic.

\subsection{Instance Selection}
We used the instances by \cite{charfreitag2024separator}, which is the most recent study on \mc-preprocessing. These are separated into four sets: \emph{easy}, \emph{medium}, \emph{big} and \emph{torus}. The first three sets contain instances from image segmentation and VLSI-chip design problems from \cite{dunning2018works} and biological, co-author and social networks from \cite{nr}\footnote{We removed the web-google instance  used in \cite{charfreitag2024separator}, as it no longer appears to be on the network repository. We instead added a larger instance called web-google to the big set.}. The instances are split into \emph{easy}, \emph{medium} and \emph{hard} based on how challanging they are for \mc. As the difficulty for \mkc\@ varies strongly depending on $k$ we decided to keep the grouping of \cite{charfreitag2024separator}. The \emph{torus} instances contain $2$ and $3$d toroid grids from statistical physics from \cite{BiqMacLib}, originally generated by \cite{liers2004computing}. We also relaxed the CELAR, GRAPH and DUTtest1 sets of frequency assignment problems from \cite{fapweb2000}, referred to as \emph{fap} instances. These were originally provided by the CALMA project and \cite{Be95GRAPH}. See Appendix \ref{appendix:instance_selection} for more detail on our instance sets.

\subsection{Running an Experiment}
We compare runs with no preprocessing, naive preprocessing (only low degree vertex removal and splitting (bi)connected components), and our preprocessing. We evaluate the $k$-values $\cbrace{3,\ldots,8,10,12}$. On the easy, medium, fap and torus datasets we employ Gurobi with a $30$ minute time limit to solve the instances to optimality. For the big data set we run a local search heuristic with a two hour time limit. Preprocessing is included in the time limit and reported runtime. To account for variances, we run $5$ seeds for Gurobi and $3$ for the heuristic.

Given a combination of instance and $k$, we determine a winner: The winner is the configuration that found the optimum solution on the largest number of seeded runs. In case of a tie, we first break by the average dual bound on the solutions not solved to optimality and then by runtime on the solved instances. 

\subsection{Computational Results}
Figure \ref{plot:solved_instances_easy_medium} shows the number of optimum solutions found for the instances from \cite{charfreitag2024separator} and the fap instance set. 
We observe significant improvements for the challenging fap data set, with the exception for $k=5$. For $k = 7$ we solve three additional instances and for $k = 8$ we solve five more instances. In both cases the additional instances solved are from the real world CELAR data set, while our performance on the synthetic GRAPH and DUTtest1 instances is the same as both naive and no preprocessing. The improvements for larger values of $k$ are due to the increased effectiveness of our \tlcname.
With increasing $k$, many of the subgraphs split by our preprocessing are easy to solve, since they can be $k$-coloured easily. 
On the torus instances, the number of solved instances decreases with increasing $k$, as the size of the search space grows significantly, while the grid structure means very few cutting planes are applicable. 
Most easy and medium instances can be routinely solved, both with and without preprocessing. Since  solution times for the easy instances are below one second even with the naive preprocessing, we do not discuss them further here (see  Appendix \ref{appendix:easy}).

\begin{figure}[h]
    \centering
    \begin{subfigure}[b]{\textwidth} 
        \centering
        \includegraphics[width=0.7\textwidth]{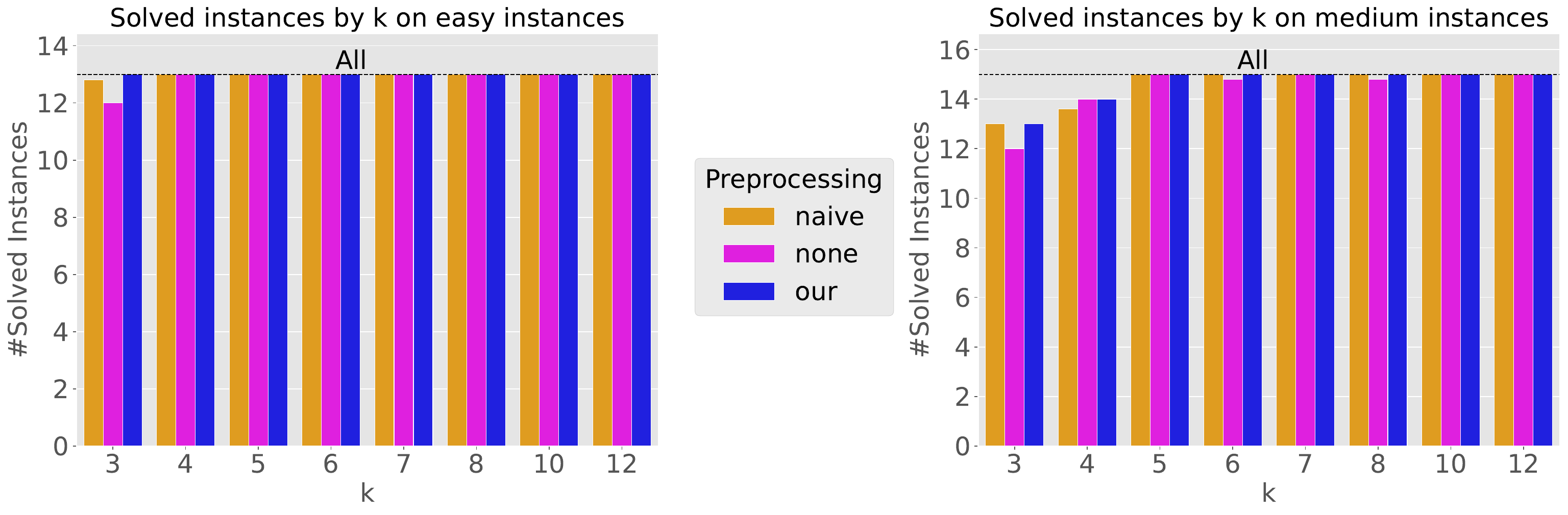}
    \end{subfigure}
    \\
    \begin{subfigure}[b]{\textwidth} 
        \centering
        \includegraphics[width=0.7\textwidth]{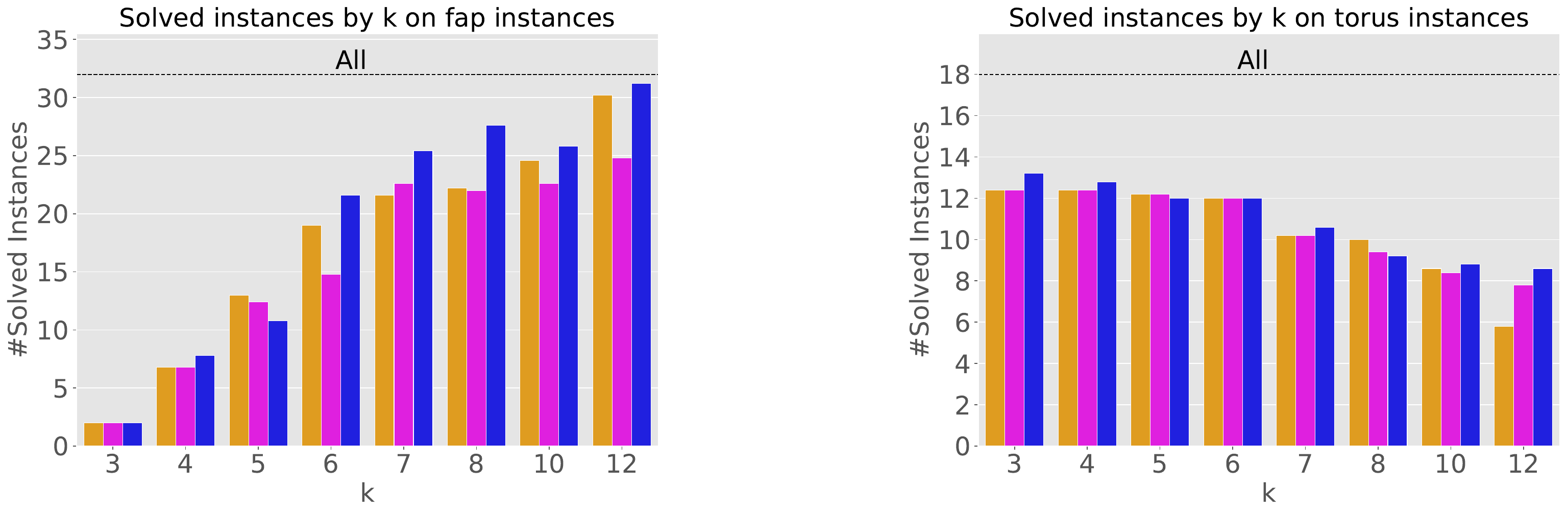}
    \end{subfigure}
    \caption{Number of optimum solutions found within $30$ minutes on \textbf{easy}, \textbf{medium}, \textbf{fap} and \textbf{torus} instances. If the optimum was found only for some seeds the instance is counted fractionally.}
    \label{plot:solved_instances_easy_medium}
\end{figure}

We now take a closer look into the medium dataset. Table \ref{tab:benchmark_medium} shows the number of wins, the runtimes, and the sizes of the remaining graphs after preprocessing. We observe significant gains over naive preprocessing. Across all $k$, the average graph kernels remaining after our preprocessing are significantly smaller than those from naive preprocessing. This leads to our preprocessing achieving more wins and lower average runtime for most $k$. This is especially visible for $k \geq 8$ where Gurobi with our preprocessing beats pure Gurobi by an order of magnitude and naive preprocessing by a factor of $6$. The most challenging instances in the medium data set are network instances. These frequently contain small, dense subgraphs that are only sparsely connected to the remaining graph. Our preprocessing splits these off or removes them with the clique, \tlcname\@ and biconnector rules, making the kernels much easier to solve.

\begin{table}[h!]
\begin{scriptsize}
\begin{center}
    {\tabcolsep2.5pt
\begin{tabular*}{\textwidth}{l@{\extracolsep{\fill}}rr@{\extracolsep{\fill}}rrrr@{\extracolsep{\fill}}rrrr}
\toprule
 & \multicolumn{2}{c}{\textbf{Gurobi}} & \multicolumn{4}{c}{\textbf{Naive}} & \multicolumn{4}{c}{\textbf{Our}} \\
\cmidrule(){2-3}\cmidrule(){4-7}\cmidrule(){8-11}
$k$ & wins & t[s] & $|V|[\%]$ & $|E|[\%]$ & wins & t[s] & $|V|[\%]$ & $|E|[\%]$ & wins & t[s] \\
\midrule
3 & 2 & 486 & 32.55 & 36.74 & 4 & 441 & \textbf{16.75} & \textbf{22.60} & \textbf{11} & \textbf{330} \\
4 & 5 & \textbf{282} & 19.77 & 21.73 & 6 & 364 & \textbf{8.34} & \textbf{12.35} & \textbf{8} & 366 \\
5 & 2 & 175 & 14.47 & 15.36 & 8 & 150 & \textbf{4.98} & \textbf{7.31} & \textbf{9} & \textbf{122} \\
6 & 1 & 143 & 11.62 & 11.53 & 9 & 47 & \textbf{3.36} & \textbf{6.05} & \textbf{11} & \textbf{40} \\
7 & 1 & 135 & 10.18 & 9.21 & \textbf{9} & 61 & \textbf{2.29} & \textbf{4.20} & \textbf{9} & \textbf{20} \\
8 & 1 & 134 & 9.02 & 7.34 & 8 & 67 & \textbf{1.55} & \textbf{3.07} & \textbf{13} & \textbf{10} \\
10 & 1 & 109 & 8.14 & 5.66 & \textbf{10} & 67 & \textbf{0.91} & \textbf{1.82} & \textbf{10} & \textbf{10} \\
12 & 1 & 119 & 7.40 & 3.91 & \textbf{11} & 67 & \textbf{0.21} & \textbf{0.27} & 10 & \textbf{10} \\
\bottomrule
\end{tabular*}
}
\end{center}
\end{scriptsize}
    \caption{Number of wins, average runtime and average percentage of vertices and edges remaining after \textbf{naive} and \textbf{our} preprocessing algorithm on the \textbf{medium} data set.}
    \label{tab:benchmark_medium}
\end{table}

On the fap dataset, preprocessing removes significantly less than on the previous data sets 
(see Table~\ref{tab:benchmark_fap}). This is primarily caused by fap instances being denser than the previous network instances. Still, our preprocessing is able to separate or remove dense subgraphs that are only sparsely connected to the remaining graph. Due to this, we observe more wins and faster runtimes on average, in addition to solving more instances to optimality.

\begin{table}[h!]
\begin{scriptsize}
\begin{center}
    {\tabcolsep2.5pt
\begin{tabular*}{\textwidth}{l@{\extracolsep{\fill}}rr@{\extracolsep{\fill}}rrrr@{\extracolsep{\fill}}rrrr}
\toprule
 & \multicolumn{2}{c}{\textbf{Gurobi}} & \multicolumn{4}{c}{\textbf{Naive}} & \multicolumn{4}{c}{\textbf{Our}} \\
\cmidrule(){2-3}\cmidrule(){4-7}\cmidrule(){8-11}
$k$ & wins & t[s] & $|V|[\%]$ & $|E|[\%]$ & wins & t[s] & $|V|[\%]$ & $|E|[\%]$ & wins & t[s] \\
\midrule
3 & 11 & 1,722 & 99.09 & 99.74 & 13 & 1,722 & \textbf{93.93} & \textbf{96.70} & \textbf{19} & \textbf{1,717} \\
4 & \textbf{13} & 1,548 & 95.20 & 98.05 & 11 & 1,535 & \textbf{91.95} & \textbf{95.73} & \textbf{13} & \textbf{1,502} \\
5 & 9 & 1,326 & 87.60 & 93.02 & \textbf{14} & \textbf{1,238} & \textbf{82.52} & \textbf{89.34} & 13 & 1,361 \\
6 & 5 & 1,020 & 63.98 & 76.38 & 10 & 760 & \textbf{56.73} & \textbf{68.66} & \textbf{18} & \textbf{627} \\
7 & 3 & 658 & 59.05 & 72.03 & 10 & 595 & \textbf{46.02} & \textbf{55.94} & \textbf{19} & \textbf{467} \\
8 & 2 & 698 & 48.02 & 62.83 & 13 & 562 & \textbf{34.72} & \textbf{43.75} & \textbf{17} & \textbf{331} \\
10 & 3 & 635 & 35.09 & 50.01 & 12 & 480 & \textbf{17.79} & \textbf{25.92} & \textbf{17} & \textbf{402} \\
12 & 0 & 475 & 27.05 & 41.25 & 12 & 126 & \textbf{11.34} & \textbf{17.59} & \textbf{25} & \textbf{66} \\
\bottomrule
\end{tabular*}
}
\end{center}
\end{scriptsize}
    \caption{Number of wins, average runtime and average percentage of vertices and edges remaining after \textbf{naive} and \textbf{our} preprocessing algorithm on the \textbf{fap} data set.}
    \label{tab:benchmark_fap}
\end{table}

The torus instances differ in their structure significantly from the other data sets as their grid structure and the large number of negative edge weights makes almost all our preprocessing rules not applicable. In Table \ref{tab:benchmark_physics} we see that naive preprocessing is able to remove almost nothing, even for higher values of $k$. Our preprocessing is able to remove significantly more, as contracting negative-weight dominating edges breaks up the grid like structure, enabling other preprocessing rules.

\begin{table}[h!]
\begin{scriptsize}
\begin{center}
    {\tabcolsep2.5pt
\begin{tabular*}{\textwidth}{l@{\extracolsep{\fill}}rr@{\extracolsep{\fill}}rrrr@{\extracolsep{\fill}}rrrr}
\toprule
 & \multicolumn{2}{c}{\textbf{Gurobi}} & \multicolumn{4}{c}{\textbf{Naive}} & \multicolumn{4}{c}{\textbf{Our}} \\
\cmidrule(){2-3}\cmidrule(){4-7}\cmidrule(){8-11}
$k$ & wins & t[s] & $|V|[\%]$ & $|E|[\%]$ & wins & t[s] & $|V|[\%]$ & $|E|[\%]$ & wins & t[s] \\
\midrule
3 & 1 & 626 & 100.00 & 100.00 & 3 & 625 & \textbf{91.40} & \textbf{95.30} & \textbf{15} & \textbf{588} \\
4 & 2 & 638 & 100.00 & 100.00 & 5 & 638 & \textbf{91.38} & \textbf{95.26} & \textbf{12} & \textbf{618} \\
5 & 5 & 671 & 97.27 & 94.75 & 6 & 669 & \textbf{84.41} & \textbf{86.60} & \textbf{10} & \textbf{653} \\
6 & \textbf{8} & 748 & 97.27 & 94.75 & 4 & \textbf{727} & \textbf{81.60} & \textbf{83.52} & \textbf{8} & 740 \\
7 & 3 & 875 & 96.25 & 92.75 & 5 & 873 & \textbf{77.14} & \textbf{77.67} & \textbf{10} & \textbf{861} \\
8 & 5 & 953 & 96.25 & 92.75 & 4 & \textbf{900} & \textbf{74.92} & \textbf{75.39} & \textbf{9} & 904 \\
10 & 7 & 1,045 & 96.25 & 92.75 & 3 & 1,014 & \textbf{72.93} & \textbf{73.17} & \textbf{8} & \textbf{944} \\
12 & 6 & 1,208 & 96.25 & 92.75 & 3 & 1,315 & \textbf{72.62} & \textbf{72.71} & \textbf{9} & \textbf{1,029} \\
\bottomrule
\end{tabular*}

}
\end{center}
\end{scriptsize}
    \caption{Number of wins, average runtime and average percentage of vertices and edges remaining after \textbf{naive} and \textbf{our} preprocessing algorithm on the \textbf{torus} data set.}
    \label{tab:benchmark_physics}
\end{table}

On the big data set we observe more mixed results. As seen in Table~\ref{tab:benchmark_big}, for most instances our preprocessing does not achieve significantly greater reductions than naive preprocessing, while taking much longer. A notable exception is the web-it-2004 instance, on which our preprocessing removes almost 90\% of nodes and 65\% of edges for all $k$, while naive preprocessing can not remove more than 12\% of nodes and 2\% of edges, even for $k = 10$. Unfortunately our preprocessing usually does not lead to better heuristic values. A significant amount of time is spend to ensure everything is optimality-preserving. This conflicts with finding the best solution within the time limit, as the objective value on the kernels has the greatest impact on the objective for the whole graph.

\begin{table}[h!]
\begin{scriptsize}
\begin{center}
    {\tabcolsep2.5pt
\begin{tabular*}{\textwidth}{ll@{\extracolsep{\fill}}r@{\extracolsep{\fill}}rrrr@{\extracolsep{\fill}}rrrr}
\toprule
 & & \textbf{Heuristic} & \multicolumn{4}{c}{\textbf{Naive}} & \multicolumn{4}{c}{\textbf{Our}} \\
\cmidrule(){4-7}\cmidrule(){8-11}
$k$ & \textbf{instance} & best value & $|V|[\%]$ & $|E|[\%]$ & bvi & pr[s] & $|V|[\%]$ & $|E|[\%]$ & bvi & pr[s] \\
\midrule 
 \multirow{ 6 }{*}{ 4 } & ca-IMDB & 3.687.347 & \textbf{39.58} & \textbf{79.82} & 50.525 & 3 & \textbf{39.58} & \textbf{79.82} & \textbf{50.589} & 63 \\
 & ca-coauthors-dblp & 12.070.341 & 93.80 & 99.58 & \textbf{3.618} & 12 & \textbf{84.78} & \textbf{97.13} & 1.480 & 1.008 \\
 & inf-road-central & 16.930.983 & \textbf{0.00} & \textbf{0.00} & \textbf{2.430} & 17 & \textbf{0.00} & \textbf{0.00} & \textbf{2.430} & 17 \\
 & web-Stanford & 2.213.473 & 65.49 & 90.45 & \textbf{4.314} & 2 & \textbf{63.48} & \textbf{89.19} & 3.096 & 178 \\
 & web-google & 4.679.030 & 56.32 & 84.80 & \textbf{22.832} & 8 & \textbf{53.47} & \textbf{82.05} & 15.212 & 1.447 \\
 & web-it-2004 & 5.727.019 & 91.44 & 98.96 & \textbf{5.208} & 4 & \textbf{11.04} & \textbf{35.92} & \textbf{5.208} & 44 \\
\midrule 
 \multirow{ 6 }{*}{ 7 } & ca-IMDB & 3.760.031 & \textbf{26.13} & \textbf{64.50} & 15.666 & 3 & \textbf{26.13} & \textbf{64.50} & \textbf{15.719} & 49 \\
 & ca-coauthors-dblp & 13.612.054 & 86.50 & 98.60 & \textbf{4.809} & 11 & \textbf{82.01} & \textbf{97.24} & 773 & 2.286 \\
 & inf-road-central & 16.933.413 & \textbf{0.00} & \textbf{0.00} & \textbf{0} & 17 & \textbf{0.00} & \textbf{0.00} & \textbf{0} & 17 \\
 & web-Stanford & 2.282.032 & 33.52 & 69.91 & 1.436 & 1 & \textbf{33.08} & \textbf{69.48} & \textbf{1.468} & 118 \\
 & web-google & 4.969.060 & 36.46 & 67.09 & 8.982 & 6 & \textbf{34.98} & \textbf{65.17} & \textbf{9.412} & 1.128 \\
 & web-it-2004 & 6.434.136 & 90.01 & 98.66 & 3.165 & 4 & \textbf{10.65} & \textbf{35.82} & \textbf{3.166} & 41 \\
\midrule 
 \multirow{ 6 }{*}{ 10 } & ca-IMDB & 3.777.961 & \textbf{16.29} & \textbf{46.38} & \textbf{3.865} & 2 & \textbf{16.29} & \textbf{46.38} & 3.858 & 37 \\
 & ca-coauthors-dblp & 14.192.268 & 79.89 & 97.19 & \textbf{4.717} & 10 & \textbf{77.07} & \textbf{96.21} & 2.338 & 1.793 \\
 & inf-road-central & 16.933.413 & \textbf{0.00} & \textbf{0.00} & \textbf{0} & 16 & \textbf{0.00} & \textbf{0.00} & \textbf{0} & 16 \\
 & web-Stanford & 2.298.360 & 23.20 & 59.40 & \textbf{912} & 2 & \textbf{22.99} & \textbf{59.13} & 907 & 63 \\
 & web-google & 5.051.920 & 22.09 & 47.17 & 4.502 & 4 & \textbf{21.26} & \textbf{45.79} & \textbf{4.768} & 561 \\
 & web-it-2004 & 6.711.894 & 88.01 & 98.02 & 2.149 & 4 & \textbf{10.26} & \textbf{35.64} & \textbf{2.150} & 42 \\
\bottomrule
\end{tabular*}
}
\end{center}
\end{scriptsize}
    \caption{Best value found by our local search heuristic within $2$ hours, percentage of vertices and edges remaining, improvement over the straight heuristic (bvi) and time spend preprocessing (pr) of \textbf{naive} and \textbf{our} preprocessing algorithm on the \textbf{big} data set. For the full table see Appenix \ref{appendix:big_full}}
    \label{tab:benchmark_big}
\end{table}

\subsection{Ablation Study}
To evaluate the impact of our different preprocessing rules, we benchmarked how the performance changes when a rule is disabled. Table \ref{tab:ablation} shows the results. Across all values of $k$, disabling the dominating edge rules leads to the largest kernels on the torus set. However, disabling the biconnector, \tlcshort\@ or \tlcshort S rule has a stronger impact on the number of wins. This suggests that these rules remove data that the solver would usually struggle with. On the medium instances, disabling the clique and dominating edge rule causes the greatest loss in the number of removed edges. This is different on the fap instances, where disabling the biconnector and the \tlcshort S rules lead to the lowest number of removed edges. However, on the medium and fap set, none of clique, biconnector or \tlcshort\@ cause significantly fewer wins when removed. This suggests that these rules have comparable impacts on our ability to solve instances to optimality.

\begin{table}[h!]
\begin{scriptsize}
\begin{center}
    {\tabcolsep2.5pt
\begin{tabular*}{\textwidth}{ll@{\extracolsep{\fill}}rr@{\extracolsep{\fill}}rr@{\extracolsep{\fill}}rr@{\extracolsep{\fill}}rr@{\extracolsep{\fill}}rr@{\extracolsep{\fill}}rr}
\toprule
 & & \textbf{All} & \textbf{Naive} & \multicolumn{2}{c}{\textbf{NoDom}} & \multicolumn{2}{c}{\textbf{NoClq}} & \multicolumn{2}{c}{\textbf{noBicon}} & \multicolumn{2}{c}{\textbf{No\tlcshort}} & \multicolumn{2}{c}{\textbf{No\tlcshort S}} \\
\cmidrule(){3-4}\cmidrule(){5-6}\cmidrule(){7-8}\cmidrule(){9-10}\cmidrule(){11-12}\cmidrule(){13-14}
$k$ & dataset & $|E|[\%]$ & $|E|[\%]$ & $|E|[\%]$ & wins & $|E|[\%]$ & wins & $|E|[\%]$ & wins & $|E|[\%]$ & wins & $|E|[\%]$ & wins  \\
\midrule 
 \multirow{ 4 }{*}{ 4 } & easy & 5.77 & 5.77 & \textbf{5.77} & \textbf{7} & \textbf{5.77} & 10 & \textbf{5.77} & 9 & \textbf{5.77} & 9 & \textbf{5.77} & 10 \\
 & medium & 12.35 & 21.73 & 14.28 & 5 & \textbf{16.52} & \textbf{4} & 13.91 & 7 & 12.32 & 5 & 12.35 & 6 \\
 & torus & 95.26 & 100.00 & \textbf{100.00} & 6 & 95.26 & 3 & 95.26 & 4 & 95.26 & \textbf{2} & 95.26 & 3 \\
 & fap & 95.73 & 98.05 & 95.73 & 11 & 96.22 & 10 & \textbf{96.60} & \textbf{6} & 95.09 & 7 & 95.82 & 14 \\
\midrule 
 \multirow{ 4 }{*}{ 7 } & easy & 0.00 & 0.00 & \textbf{0.00} & \textbf{11} & \textbf{0.00} & 12 & \textbf{0.00} & \textbf{11} & \textbf{0.00} & 12 & \textbf{0.00} & 12 \\
 & medium & 4.20 & 9.21 & \textbf{6.60} & 10 & 5.48 & 7 & 4.71 & 11 & 4.20 & \textbf{6} & 4.20 & 9 \\
 & torus & 77.67 & 92.75 & \textbf{87.10} & 8 & 77.67 & 2 & 78.42 & 4 & 77.35 & 3 & 79.95 & \textbf{1} \\
 & fap & 55.94 & 72.03 & 55.94 & \textbf{5} & 59.64 & 6 & 55.91 & 6 & 56.37 & 7 & \textbf{67.07} & 12 \\
\midrule 
 \multirow{ 4 }{*}{ 10 } & easy & 0.00 & 0.00 & \textbf{0.00} & 12 & \textbf{0.00} & \textbf{11} & \textbf{0.00} & \textbf{11} & \textbf{0.00} & 12 & \textbf{0.00} & 12 \\
 & medium & 1.82 & 5.66 & \textbf{4.31} & 10 & 2.13 & \textbf{9} & 2.30 & 10 & 1.88 & \textbf{9} & 1.82 & 11 \\
 & torus & 73.17 & 92.75 & \textbf{84.18} & 9 & 73.17 & \textbf{2} & 74.12 & \textbf{2} & 72.68 & 3 & 76.63 & \textbf{2} \\
 & fap & 25.92 & 50.01 & 25.92 & 6 & 37.08 & \textbf{3} & 25.92 & 4 & 23.03 & 4 & \textbf{40.12} & 18 \\
\bottomrule
\end{tabular*}

}
\end{center}
\end{scriptsize}
    \caption{Percentage of edges remaining and number of wins comparing full preprocessing except dominating edges (\textbf{NoDom}), cliques (\textbf{noClq}), $2$ vertex separators (\textbf{noBicon}), \tlcname\ (\textbf{no\tlcshort}) and \tlcname\ with solving the small side (\textbf{no\tlcshort S}). Naive and full preprocessing for reference, the worst values are highlighted. For the full table see Appendix \ref{appendix:ablation_full}}
    \label{tab:ablation}
\end{table}
\section{Conclusion and Outlook}
We introduced new optimality-preserving data reduction and data separation rules for \mkc. In particular, we proposed \tlcname, a new data separation rule that is especially effective for $k\ge 3$. Our experimental results show that \tlcname\@ enables the removal of many cut sets, particularly for larger $k$. 
Combined with generalisations of \mc-preprocessing rules, this allows us to solve significantly more instances than using an exact solver with naive preprocessing only. 
This holds for instances from different real-world applications and across different values of $k$.

Future work includes developing additional criteria for cut sets whose removal preserves optimality. It would also be interesting to investigate whether \tlcname\ can be used for problems related to \mkc. Furthermore, it would be interesting to engineer a preprocessing framework especially suited for running heuristics on very large instances.



\bibliography{sources}

@article{kaibel2011orbitopal,
  title={{Orbitopal fixing}},
  author={Kaibel, Volker and Peinhardt, Matthias and Pfetsch, Marc E},
  journal={Discrete Optimization},
  volume={8},
  number={4},
  pages={595--610},
  year={2011},
  publisher={Elsevier},
  doi={10.1016/j.disopt.2011.07.001}
}

@article{chopra1993partition,
  title={{The partition problem}},
  author={Chopra, Sunil and Rao, Mendu R},
  journal={Mathematical programming},
  volume={59},
  number={1},
  pages={87--115},
  year={1993},
  publisher={Springer},
  doi={10.1007/BF01581239}
}

@article{chopra1995facets,
  title={{Facets of the k-partition polytope}},
  author={Chopra, Sunil and Rao, Mendu R},
  journal={Discrete applied mathematics},
  volume={61},
  number={1},
  pages={27--48},
  year={1995},
  publisher={Elsevier},
  doi={10.1016/0166-218X(93)E0175-X}
}

@article{ghaddar2011branch,
  title={{A branch-and-cut algorithm based on semidefinite programming for the minimum k-partition problem}},
  author={Ghaddar, Bissan and Anjos, Miguel F and Liers, Frauke},
  journal={Annals of Operations Research},
  volume={188},
  number={1},
  pages={155--174},
  year={2011},
  publisher={Springer},
  doi={10.1007/s10479-008-0481-4}
}

@incollection{anjos2013solving,
  title={{Solving k-way graph partitioning problems to optimality: The impact of semidefinite relaxations and the bundle method}},
  author={Anjos, Miguel F and Ghaddar, Bissan and Hupp, Lena and Liers, Frauke and Wiegele, Angelika},
  booktitle={Facets of combinatorial optimization: Festschrift for Martin Gr{\"o}tschel},
  pages={355--386},
  year={2013},
  publisher={Springer},
  doi={10.1007/978-3-642-38189-8_15}
}

@article{van2016new,
  title={{New bounds for the max-k-cut and chromatic number of a graph}},
  author={van Dam, Edwin R and Sotirov, Renata},
  journal={Linear Algebra and its Applications},
  volume={488},
  pages={216--234},
  year={2016},
  publisher={Elsevier},
  doi={10.1016/j.laa.2015.09.043}
}

@article{de2019improving,
  title={{Improving the linear relaxation of maximum k-cut with semidefinite-based constraints}},
  author={Rodrigues de Sousa, Vilmar Jeft{\'e} and Anjos, MiguelF and Le Digabel, S{\'e}bastien},
  journal={EURO Journal on Computational Optimization},
  volume={7},
  number={2},
  pages={123--151},
  year={2019},
  publisher={Elsevier},
  doi={10.1007/s13675-019-00110-y}
}

@article{rodrigues2018computational,
  title={{Computational study of valid inequalities for the maximum k-cut problem}},
  author={Rodrigues de Sousa, Vilmar Jeft{\'e} and Anjos, Miguel F and Le Digabel, S{\'e}bastien},
  journal={Annals of Operations Research},
  volume={265},
  number={1},
  pages={5--27},
  year={2018},
  publisher={Springer},
  doi={10.1007/s10479-017-2448-9}
}

@article{ales2016extended,
  title={{An extended edge-representative formulation for the k-partitioning problem}},
  author={Ales, Zacharie and Knippel, Arnaud},
  journal={Electronic Notes in Discrete Mathematics},
  volume={52},
  pages={333--342},
  year={2016},
  doi={10.1016/j.endm.2016.03.044}
}

@article{de2022computational,
  title={{Computational study of a branching algorithm for the maximum k-cut problem}},
  author={Rodrigues de Sousa, Vilmar Jeft{\'e} and Anjos, Miguel F and Le Digabel, S{\'e}bastien},
  journal={Discrete Optimization},
  volume={44},
  pages={100656},
  year={2022},
  publisher={Elsevier},
  doi={10.1016/j.disopt.2021.100656}
}

@article{ma2017multiple,
  title={{A multiple search operator heuristic for the max-k-cut problem}},
  author={Ma, Fuda and Hao, Jin-Kao},
  journal={Annals of Operations Research},
  volume={248},
  number={1},
  pages={365--403},
  year={2017},
  publisher={Springer},
  doi={10.1007/s10479-016-2234-0}
}

@inproceedings{charfreitag2024separator,
  title={{Separator Based Data Reduction for the Maximum Cut Problem}},
  author={Charfreitag, Jonas and Dahn, Christine and Kaibel, Michael and Mayer, Philip and Mutzel, Petra and Sch{\"u}rmann, Lukas},
  booktitle={22nd International Symposium on Experimental Algorithms (SEA 2024)},
  pages={4--1},
  year={2024},
  organization={Schloss Dagstuhl--Leibniz-Zentrum f{\"u}r Informatik},
  doi={10.4230/LIPIcs.SEA.2024.4}
}

@inproceedings{ferizovic2020engineering,
  title={{Engineering kernelization for maximum cut}},
  author={Ferizovic, Damir and Hespe, Demian and Lamm, Sebastian and Mnich, Matthias and Schulz, Christian and Strash, Darren},
  booktitle={2020 Proceedings of the Twenty-Second Workshop on Algorithm Engineering and Experiments (ALENEX)},
  pages={27--41},
  year={2020},
  organization={SIAM},
  doi={10.1137/1.9781611976007.3}
}

@article{rehfeldt2023faster,
  title={{Faster exact solution of sparse MaxCut and QUBO problems}},
  author={Rehfeldt, Daniel and Koch, Thorsten and Shinano, Yuji},
  journal={Mathematical Programming Computation},
  volume={15},
  number={3},
  pages={445--470},
  year={2023},
  publisher={Springer},
  doi={10.1007/s12532-023-00236-6}
}

@article{fakhimi2025folding,
  title={A folding preprocess for the max k-cut problem},
  author={Fakhimi, Ramin and Validi, Hamidreza and Hicks, Illya V and Terlaky, Tam{\'a}s and Zuluaga, Luis F},
  journal={Optimization Letters},
  volume={19},
  number={5},
  pages={899--917},
  year={2025},
  publisher={Springer},
  doi={10.1007/s11590-025-02199-0}
}

@article{hopcroft1973dividing,
  title={{Dividing a graph into triconnected components}},
  author={Hopcroft, John E and Tarjan, Robert Endre},
  journal={SIAM Journal on computing},
  volume={2},
  number={3},
  pages={135--158},
  year={1973},
  publisher={SIAM},
  doi={10.1137/0202012}
}

@inproceedings{gutwenger2000linear,
  title={{A linear time implementation of SPQR-trees}},
  author={Gutwenger, Carsten and Mutzel, Petra},
  booktitle={International Symposium on Graph Drawing},
  pages={77--90},
  year={2000},
  organization={Springer},
  doi={10.1007/3-540-44541-2_8}
}

@article{chimani2019cut,
  title={{Cut polytopes of minor-free graphs}},
  author={Chimani, Markus and Juhnke-Kubitzke, Martina and Nover, Alexander and R{\"o}mer, Tim},
  journal={arXiv preprint arXiv:1903.01817},
  year={2019},
  doi={10.48550/arXiv.1903.01817}
}

@article{hochbaum1993should,
  title={{Why should biconnected components be identified first}},
  author={Hochbaum, Dorit S},
  journal={Discrete applied mathematics},
  volume={42},
  number={2-3},
  pages={203--210},
  year={1993},
  publisher={Elsevier},
  doi={10.1016/0166-218X(93)90046-Q}
}

@article{hopcroft1973n,
  title={{An n\^{}5/2 algorithm for maximum matchings in bipartite graphs}},
  author={Hopcroft, John E and Karp, Richard M},
  journal={SIAM Journal on computing},
  volume={2},
  number={4},
  pages={225--231},
  year={1973},
  publisher={SIAM},
  doi={10.1137/0202019}
}

@article{tinhofer1984probabilistic,
  title={{A probabilistic analysis of some greedy cardinality matching algorithms}},
  author={Tinhofer, Gottfried},
  journal={Annals of Operations Research},
  volume={1},
  number={3},
  pages={239--254},
  year={1984},
  publisher={Springer},
  doi={10.1007/BF01874391}
}

@article{karger1996new,
  title={{A new approach to the minimum cut problem}},
  author={Karger, David R and Stein, Clifford},
  journal={Journal of the ACM (JACM)},
  volume={43},
  number={4},
  pages={601--640},
  year={1996},
  publisher={ACM New York, NY, USA},
  doi={10.1145/234533.234534}
}

@inproceedings{lange2019combinatorial,
  title={{Combinatorial persistency criteria for multicut and max-cut}},
  author={Lange, Jan-Hendrik and Andres, Bjoern and Swoboda, Paul},
  booktitle={Proceedings of the IEEE/CVF Conference on Computer Vision and Pattern Recognition},
  pages={6093--6102},
  year={2019},
  doi={10.1109/CVPR.2019.00625}
}

@book{eisenblatter2002frequency,
  title={{Frequency assignment in GSM networks: Models, heuristics and lower bounds}},
  author={Eisenbl{\"a}tter, Andreas},
  year={2002},
  publisher={Cuvillier Verlag},
  url={https://nbn-resolving.org/urn:nbn:de:0297-zib-10132}
}

@article{dahl2007integer,
  title={{An integer programming approach to image segmentation and reconstruction problems}},
  author={Dahl, Geir and Flatberg, Truls},
  journal={Geometric Modelling, Numerical Simulation, and Optimization: Applied Mathematics at SINTEF},
  pages={475--496},
  year={2007},
  publisher={Springer},
  doi={10.1007/978-3-540-68783-2_14}
}

@article{hendrickson1995improved,
  title={{An improved spectral graph partitioning algorithm for mapping parallel computations}},
  author={Hendrickson, Bruce and Leland, Robert},
  journal={SIAM Journal on Scientific Computing},
  volume={16},
  number={2},
  pages={452--469},
  year={1995},
  publisher={SIAM},
  doi={10.1137/0916028}
}

@article{walshaw1997parallel,
  title={{Parallel dynamic graph partitioning for adaptive unstructured meshes}},
  author={Walshaw, Chris and Cross, Mark and Everett, Martin G},
  journal={Journal of Parallel and Distributed Computing},
  volume={47},
  number={2},
  pages={102--108},
  year={1997},
  publisher={Elsevier},
  doi={10.1006/jpdc.1997.1407}
}

@inproceedings{Karp1972,
  author       = {Richard M. Karp},
  editor       = {Raymond E. Miller and
                  James W. Thatcher},
  title        = {{Reducibility Among Combinatorial Problems}},
  booktitle    = {Proceedings of a symposium on the Complexity of Computer Computations,
                  held March 20-22, 1972, at the {IBM} Thomas J. Watson Research Center,
                  Yorktown Heights, New York, {USA}},
  series       = {The {IBM} Research Symposia Series},
  pages        = {85--103},
  publisher    = {Plenum Press, New York},
  year         = {1972},
  bibsource    = {dblp computer science bibliography, https://dblp.org},
  doi          = {10.1007/978-1-4684-2001-2\_9},
}

@inproceedings{nr,
      title = {{The Network Data Repository with Interactive Graph Analytics and Visualization}},
      author={Ryan A. Rossi and Nesreen K. Ahmed},
      booktitle = {AAAI},
      url={https://networkrepository.com},
      year={2015}
}

@online{fapweb2000,
  author  = {Eisenblätter, Andreas and Koster, Arie},
  title   = {{FAP web}},
  year    = {2000},
  url     = {https://fap.zib.de/},
  publisher = {Zuse Institute Berlin},
  urldate = {2026-04-08}
}

@MastersThesis{Be95GRAPH,
  author       = "H. P. {v}an Benthem",
  title        = "{GRAPH Generating Radio Link Frequency Assignment Problems Heuristically}",
  school       = "Delft University of Technology",
  year         = "1995",
  language     = "English",
}

@Misc{BiqMacLib,
  author       = {Angelika Wiegele},
  howpublished = {\url{http://biqmac.aau.at/biqmaclib.html}},
  title        = {Biq {M}ac {L}ibrary -- A collection of {M}ax-{C}ut and quadratic 0-1 programming instances of medium size},
  year         = {2009},
}

@article{liers2004computing,
  title={{Computing exact ground states of hard Ising spin glass problems by branch-and-cut}},
  author={Liers, Frauke and J{\"u}nger, Michael and Reinelt, Gerhard and Rinaldi, Giovanni},
  journal={New optimization algorithms in physics},
  pages={47--69},
  year={2004},
  publisher={Wiley Online Library},
  doi={10.1002/3527603794.ch4}
}

@article{dunning2018works,
  title={{What works best when? A systematic evaluation of heuristics for Max-Cut and QUBO}},
  author={Dunning, Iain and Gupta, Swati and Silberholz, John},
  journal={INFORMS Journal on Computing},
  volume={30},
  number={3},
  pages={608--624},
  year={2018},
  publisher={INFORMS},
  doi={10.1287/ijoc.2017.0798}
}

@misc{gurobi,
  author = {{Gurobi Optimization, LLC}},
  title = {{Gurobi Optimizer Reference Manual}},
  year = 2024,
  url = "https://www.gurobi.com"
}

@article{staudt2016networkit,
  title={{NetworKit: A tool suite for large-scale complex network analysis}},
  author={Staudt, Christian L and Sazonovs, Aleksejs and Meyerhenke, Henning},
  journal={Network Science},
  volume={4},
  number={4},
  pages={508--530},
  year={2016},
  doi={10.1017/nws.2016.20}
}

@article{chimani2013open,
  title={{The Open Graph Drawing Framework (OGDF).}},
  author={Chimani, Markus and Gutwenger, Carsten and J{\"u}nger, Michael and Klau, Gunnar W and Klein, Karsten and Mutzel, Petra},
  journal={Handbook of graph drawing and visualization},
  volume={2011},
  pages={543--569},
  year={2014}
}

\appendix

\section{Preprocessing Framework}

\subsection{Offsets} \label{appendix:offset_overview}
The constant offsets of the discussed rules are as follows:
\begin{itemize}
    \item \textbf{Low degree}: $\sum_{e \in \delta(v)} w(e)$
    \item \textbf{Cliques with a small neighbourhood}: Let $c$ be the unit edge weight of $C$ and $F = N(V_C) \cup V_C$. Then the offset is
    \begin{align*}
        \frac{c}{2} \cdot \left( |F|^2 - (k - (|F| \mod k)) \cdot \left\lfloor \frac{|F|}{k} \right\rfloor^2 - (|F| \mod k) \cdot \left\lceil \frac{|F|}{k} \right\rceil^2 \right)
    \end{align*}
    \item \textbf{$2$-Vertex Separators}: $w(p')$
    \item \textbf{Negative Dominating Edges}: $0$
    \item \textbf{Negative Dominating Triangles} (see Appendix \ref{appendix:dominating_edges_extra}): $0$
    \item \textbf{(Bi)connected Components}: $0$
    \item \textbf{Structured Cut Sets}: $\sum_{e \in C} w(e)$
    \begin{itemize}
        \item If the optimum solution $p$ for w.l.o.g.\@ $G_2$ is known, then $w(p) + \sum_{e \in C} w(e)$.
    \end{itemize}
\end{itemize}

\subsection{Combining Partitions}\label{appendix:criteria_preserving_optimality}
Proof of Lemma \ref{lemma:solution_reconstruction}:

\begin{proof}
    To compute $p$ from $p'$ and $p''$ we use the following algorithm:

        \begin{algorithm}[H]
        \caption{\textsc{CombinePartitions}$(p': V_1 \rightarrow \{1,...,k\},p'': V_2 \rightarrow \{1,...,k\})$}
        \label{algorithm:combine_partitions}
        \begin{algorithmic}[1]
            \State Initialize array $\pi$ of length $k$ with entries $-1$
            \For{$v \in V_1 \cap V_2$}
                \State $\pi[p''(v)] = p'(v)$
            \EndFor
            \State Fill entries of $\pi$ that are still $-1$ with the numbers from $\{1,...,k\}$ that are not yet in $\pi$
            \State Compute the $k$-partition $p$ of $V_1 \cup V_2$ with
            \begin{align*}
                p(v) = \begin{cases}
                p'(v) & \text{ if } v \in V_1 \\
                \pi[p''(v)] & \text{ otherwise}
            \end{cases}
            \end{align*}
            \State \textbf{return} $p$
        \end{algorithmic}
    \end{algorithm}

    $p$ clearly satisfies $\partition{p}{V_1}{} \cong p'$. As $\partition{p'}{V_1 \cap V_2}{} \cong \partition{p''}{V_1 \cap V_2}{}$ we have that the array $\pi$ we compute stores a permutation such that $\partition{p'}{V_1 \cap V_2}{} = \partition{p''}{V_1 \cap V_2}{\pi}$. We then have for all $v \in V_2$ that $p(v) = \pi[p''(v)]$, therefore $\partition{p}{V_2}{} \cong \partition{p''}{}{}$.

    By realising partitions such that read and write access to $p(v)$ is possible in $\cO(1)$ (for example as arrays), this algorithm runs in $\cO(|V_1 \cup V_2|) = \cO(n)$.
\end{proof}
\section{Data Reduction Rules} 
\subsection{Proof of Theorem \ref{theorem:two_vertex_cut}} \label{appendix:spqr_reduction}

\begin{proof}
    In \cite{charfreitag2024separator} this was proven using their framework based on graph separators. Our proof is very similar, but based on our proof framework based on graph addition.
    
    We assume w.l.o.g.\@ that $\{u, v\} \in E$ (and if not insert it with weight $0$). Let $G[V'] = G' = (V', E', w')$ and $\Delta = w'(p'') - w'(p')$. We now set $\hat{w}(\{u, v\}) = w'(\cbrace{u, v}) - \Delta$ and $\hat{w}(e) = w'(e) = w(e)$ for all other $e \in E'$. The graph $\hat{G} = (V', E', \hat{w})$ is $G'$, except we subtracted $\Delta$ from the weight of $\{u, v\}$. Let $G''$ be the graph that remains after the rule has been applied. Note that $G = \hat{G} + G''$.
    
    We begin by claiming that $p'$ and $p''$ are both optimum solutions for $\hat{G}$. To show this, we first show that they are the optimum solution in which $u, v$ have the same resp.\@ different colours and then show that they have equal cut value.
    
    As changing the weight of $\{u, v\}$ has no impact on the cut value of any solution that does not cut $\{u, v\}$, $p'$ remains an optimum among these. For the solutions that cut $\{u, v\}$ all of their objective values change by $-\Delta$, so $p''$ remains an optimum among them. We then get
    \begin{align*}
        \hat{w}(p'') = w'(p'') - \Delta = w'(p'') - (w'(p'') - w'(p')) = w'(p')
    \end{align*}

    Therefore, given an optimum solution $p'$ for $G'$, we either have $\partition{p'}{\{u, v\}}{} \cong \partition{p}{\{u, v\}}{}$ or $\partition{p''}{\{u, v\}}{} \cong \partition{p}{\{u, v\}}{}$. In either case we can combine an optimum solution for $G''$ with an optimum solution for $G'$ using Lemma \ref{lemma:solution_reconstruction}. By Theorem \ref{theorem:graph_sum_optimality}, the resulting partition $p$ is optimum for $G$.
\end{proof}

\subsection{Larger Vertex Separators} \label{appendix:3_vertex_cuts}
A natural question that arises is whether or not we can also use $3$ vertex separators or even larger separators for a similar preprocessing strategy. Our clique based rules work on larger separators, but rely on a very special structure of the graphs on the other side that we can not hope for in general. In \cite{charfreitag2024separator} the usage of $3$ vertex cuts was employed for \mc. Unfortunately, this can not be generalised to higher values of $k$.

For \mc, given a $3$ vertex separator $u, v, w$, the weights of $\{u, v\},\{u, w\}$ and $\{v, w\}$ could be modified. Together with the offset, this gives $4$ degrees of freedom to encode the differences in objective values for the $4$ different optimum values for ``$u, v, w$ are in the same group'', ``$u, v$ are in the same group, $w$ is in the other group'', ``$u, w$ are in the same group, $v$ is in the other group'' and ``$v, w$ are in the same group, $u$ is in the other group''. For \mkc\@ with $k \geq 3$ we would have to encode a fifth optimum value for the case ``$u, v, w$ are all in different groups'', which makes using $3$ vertex cuts in general not possible.

\subsection{Dominating Edges and Triangles} \label{appendix:dominating_edges_extra}
In \mc\@ preprocessing, Lange et. al \cite{lange2019combinatorial} discussed that an edge $e = \cbrace{u, v}$ can be contracted if we know that there is an optimum solution that cuts $e$. This is achieved by multiplying the edge weights of edges incident to $u$ before contracting $u$ into $v$. This means if $x$ and $v$ are on different sides of the cut in the reduced graph, we take a penalty of $w(\cbrace{x, u})$, as $x$ and $u$ are on the same side after reconstruction. If $x$ and $v$ are on the same side, then $x$ and $u$ are on different sides after reconstruction and we take no penalty.

For \mkc\@ for $k \geq 3$ this unfortunately does not work, as there is also the configuration $x, u$ and $v$ all have different colours. The same issue was noted by Lange et. al \cite{lange2019combinatorial} with their preprocessing for \textsc{MultiCut}.

\subsubsection{Dominating Triangles}
\begin{theorem}[Negative Triangle Rule (Corollary 1 (i) from \cite{lange2019combinatorial})] \label{theorem:triangle_mc}
    Given an weighted undirected graph $G = (V, E, w)$ and vertices $v_1, v_2, v_3 \in V$ such that $e_1 = \{v_1, v_2\}, e_2 = \{v_1, v_3\}, e_3 = \{v_2, v_3\} \in E$. Let $U_1, U_2 \subseteq V$ with $e_1, e_2 \in \delta(U_1)$ and $e_1, e_3 \in \delta(U_2)$. If $w(e_1) < 0$ and the inequalities
        \begin{align*}
            - w(e_1) - w(e_2) &\geq \sum_{e' \in \delta(U_1) \setminus \{e_1, e_2\}} |w(e')| \\
            - w(e_1) - w(e_3) &\geq \sum_{e' \in \delta(U_2) \setminus \{e_1, e_3\}} |w(e')|
        \end{align*}
        hold, contracting $e_1$ is optimality-preserving.
\end{theorem}

For the proof we refer to \cite{lange2019combinatorial}. Unfortunately the rule does not immediately extend to \mkc. Luckily we can show that a stricter version the rule still holds for \mkc.

\begin{theorem}[Negative Triangle Rule for \mkc] \label{theorem:dominating_triangle_mkc}
    Given an weighted undirected graph $G = (V, E, w)$ and vertices $v_1, v_2, v_3 \in V$ such that $e_1 = \{v_1, v_2\}, e_2 = \{v_1, v_3\}, e_3 = \{v_2, v_3\} \in E$. Let $U_1, U_2 \subseteq V$ with $e_1, e_2 \in \delta(U_1)$ and $e_1, e_3 \in \delta(U_2)$. If $w(e_1) < 0$ and the two inequalities
        \begin{align*}
            - w(e_1) &\geq \sum_{e' \in \delta(U_1) \setminus \{e_1, e_2\}} |w(e')| \\
            - w(e_1) &\geq \sum_{e' \in \delta(U_2) \setminus \{e_1, e_3\}} |w(e')|
        \end{align*}
        hold, then contracting $e$ is optimality-preserving.
\end{theorem}

\begin{proof}
    Let $p$ be a $k$-partition of $V$.

    If $p$ cuts $e_1$ we produce a $k$-partition at least as good as $p$ that does not cut $e_1$. If $p$ cuts $e_1$, then it also cuts $e_2$ or $e_3$, w.l.o.g.\@ it cuts $e_2$. We then permute the colours on $U_1$ such that $e_1$ is no longer cut. let $p'$ be the resulting partition. In the worst case $p$ cuts all positive and none of the negative edges in $\delta(U_1) \setminus \cbrace{e_1, e_2}$ and $p'$ cuts no positive edges and all negative edges in $\delta(U_1)$, except for $e_1$. The difference in objective value is $-w(e_1) - \sum_{e' \in \delta(U_1) \setminus \{e_1, e_2\}} |w(e')| \geq 0$. Therefore we find a partition at least as good as $p$ which does not cut $e_1$.
\end{proof}
\section{Proofs and Further Notes for \tlctitle}
\subsection{Proof of Theorem \ref{theorem:two_layer_reconstructor}} \label{appendix:two_layer_reconstructor}
\begin{proof}
    We begin with the runtime, as it is the easiest. Computing $G_{CR}$ takes $\cO(|C| + k)$ time and its bipartite complement can be computed in time $\cO(k^2)$. A maximum matching can be computed in $\cO(k^{2.5})$ \cite{hopcroft1973n} and, if a matching is found, all further steps take time $\cO(n + k)$, leading to a total runtime of $\cO(n + |C| + k^{2.5})$. 
    
    With this we now move on to the proof of correctness. We will begin by showing that if we find a perfect matching, we find a $k$-partition $p$ that cuts all $e \in C$ with $\partition{p}{V_1}{} \cong p'$ and $\partition{p}{V_2}{} \cong p''$. We then go on to show that we find a perfect matching if such a partition $p$ exists.

    If we find a perfect bipartite matching in line $3$, then $\pi$ defines a permutation on $\{1,...,k\}$. By construction the returned partition $p$ has $\partition{p}{V_1}{} = p'$ and $\partition{p}{V_2}{} = p''$, so we only need to show that $p$ cuts every edge $e \in C$. \\
    Let $\{v, w\} \in C$ be an edge with $v \in V_1$ and $w \in V_2$. From this it follows that the edge $\{(p'(v), L), (p''(w), R)\}$ is in $G_{CR}$ and therefore not in its bipartite complement, so it can't be in the matching. Therefore we get $\pi(p''(w)) \not = p'(v)$, so $p(v) = p'(v) \not = \pi(p''(w)) = p(w)$, so $e$ is cut by $p$.

    Now we just have to show that we find a perfect matching if a partition $p$ that cuts all $e \in C$ with $\partition{p}{V_1}{} \cong p'$ and $\partition{p}{V_2}{} \cong p''$ exists. \\
    Suppose that such a $k$-partition $p$ of $V$ exists. W.l.o.g.\@ $\partition{p}{V_1}{} = p'$ and let $\pi$ be the permutation such that $\partition{p}{V_2}{} = \partition{p''}{}{\pi}$. For better readability let $\cL = \{(i, L) \mid i \in \{1,...,k\}\}$ and $\cR = \{(i, R) \mid i \in \{1,...,k\}\}$ be the different sides of the colour relation graph and $H = \bpc{G_{CR}}{\cL, \cR}$ be the bipartite complement of $G_{CR}$. We claim that $M = \{\{(i, L), (\pi^{-1}(i), R)\}\}$ is a perfect matching for $H$. Suppose that there was an edge $\{(i, L), (\pi^{-1}(i), R)\} \in M$ that is not in $H$. Then there must be an edge $e = \{v, w\} \in C$ with $v \in V_1, w \in V_2$ such that $p'(v) = i, p''(w) = \pi^{-1}(i)$. But then $p(v) = i = \pi(\pi^{-1}(i)) = \pi(p''(i)) = p(w)$, so $p$ does not cut $e$, a contradiction. Therefore $M$ must be a perfect matching.

    With this we have shown that \textsc{Reconstruct\tlctitle} finds a $k$-partition $p$ of $V$ that cuts all $e \in C$ with $\partition{p}{V_1}{} \cong p'$ and $\partition{p}{V_2}{} \cong p''$ if one exists in time $\cO(n + |C| + k^{2.5})$. and returns an error otherwise.
\end{proof}

\subsection{Proof of Theorem \ref{theorem:no_big_two_layer_separators}} \label{appendix:no_big_two_layer}
\begin{proof}
    Let w.l.o.g.\@ $|\partial_G(V_1)| \geq k$ and let $v_1,...,v_k \in \partial_G(V_1)$ be distinct vertices. We now define $p'(v_i) = i$ and assign all other vertices in $V_1$ to some arbitrary colour. We define $p''$ with $p''(v) = 1$ for all $v \in V_2$.

    Suppose now that there was a $k$-partition $p$ of $V$ such that $p$ cuts all edges in $C$, $\partition{p}{V_1}{} \cong p_1$ and $\partition{p}{V_2}{} \cong p_2$. Then w.l.o.g.\@ we assume that $\partition{p}{V_1}{} = p_1$. Let $j = p(w)$ for all $w \in V_2$. Then there exists an edge $\{v_j, w\}$ for some $w \in V_2$ (as otherwise $v_j \not \in \partial_G(V_1)$), so $p(v_j) = j = p(w)$. Therefore $p$ does not cut $\{v_j, w\}$, a contradiction.
\end{proof}

\subsection{Proof of Theorem \ref{theorem:two_layer_sufficient}} \label{appendix:two_layer_sufficient}
\begin{proof}
    Let $p', p''$ be optimum solutions for $G_1 = (V_1, E_1, w_1), G_2 = (V_2, E_2, w_2)$ respectively.

    If $\tlc{C}$ has at most $k$ vertices, then we can permute $p''$ such that no two vertices $v \in \partial_G(V_1), w \in \partial_G(V_2)$ have the same colour and therefore all edges in $C$ are cut.

    For the case that $\tlc{C}$ has at most $k-1$ edges we show that algorithm \ref{algorithm:byDegreeAscending} always computes a perfect matching, despite only guaranteeing a $\frac{1}{2}$ approximation on general bipartite graphs.
    \begin{algorithm}
        \caption{ByDegreeAscending$(G = (L \sqcup R, E))$}
        \label{algorithm:byDegreeAscending}
        \begin{algorithmic}[1]
            \State Set $M = \emptyset$
            \State Set $G' = G$
            \State Sort $L = \{v_1,...,v_k\}$ by degree in $G$ ascending
            \For{$v_i \in L$}
                \State Choose any edge $e = \{v_i, w\}$ in $G'$ \label{algorithm:byDegreeAscending:line:pickEdge}
                \State Set $M = M \cup \{e\}$
                \State Delete $v_i$ and $w$ from $G$
            \EndFor
            \State \textbf{return} $M$
        \end{algorithmic}
    \end{algorithm}

    \begin{lemma} \label{theorem:by_degree_ascending_perfect}
        Given a bipartite graph $G = \brac{L \sqcup R, E}, |L| = |R| = k$ with $|E| \geq k^2 - k + 1$ the \textsc{ByDegereeAscending} algorithm \ref{algorithm:byDegreeAscending} computes a perfect matching.
    \end{lemma}
    \begin{proof}
        To show that the algorithm produces a perfect matching it suffices to show that we can always find an edge $e$ in line \ref{algorithm:byDegreeAscending:line:pickEdge}. Suppose that in iteration $i$ vertex $v_i$ has no more incident edges. As  every iteration before has reduced the degree of $v_i$ by at most $1$ we can conclude that $d_G(v) \leq i - 1$. As every previous iteration selected a different vertex than $v_i$ we can conclude that there at least $i$ vertices in $L$ with degree $\leq i - 1$. Therefore we have that $|E| \leq k^2 - i \cdot (k - i + 1)$ and therefore
        \begin{align*}
            i \cdot (k - i + 1) \leq k - 1 \Leftrightarrow (k - i) \cdot (i - 1) + 1 \leq 0
        \end{align*}
        However, $k \geq i$ and $i \geq 1$, therefore $(k - i) \cdot (i - 1) + 1 \geq 1$, a contradiction. Therefore we always find an edge in line \ref{algorithm:byDegreeAscending:line:pickEdge} and compute a perfect matching.
    \end{proof}

    Algorithm \ref{algorithm:byDegreeAscending} is very similar to the \textsc{GreedyMin} algorithm \cite{tinhofer1984probabilistic}, which selects a vertex of minimum degree in the remaining graph uniformly at random, matches it to a random neighbour and then removes both vertices. The question arises if \textsc{GreedyMin} also always finds a perfect matching if we are at most $k-1$ edges from a complete bipartite graph. We answer the question in the positive:
    \begin{corollary}
        Given a bipartite graph $G = \brac{L \sqcup R, E}, |L| = |R| = k$ with $|E| \geq k^2 - k + 1$ the \textsc{GreedyMin} algorithm \cite{tinhofer1984probabilistic} computes a perfect matching.
    \end{corollary}
    \begin{proof}
        Let $G_i = \brac{L_i \sqcup R_i, E_i}$ be the graph after the $i$th iteration, $G_0 = G$. We show that $G_i$ always contains $\geq (k - i)^2 - (k - i - 1)$ edges and therefore by Theorem \ref{theorem:by_degree_ascending_perfect} a perfect matching. This implies that in no iteration the minimum degree vertex has degree $0$. Hence we compute a perfect matching. We show this via induction over $i$. \\
        For $i = 0$ it immediately holds. From $i$ to $i + 1$ we distinguish two cases: If $G_i$ is complete, then $G_{i + 1}$ is also complete and therefore has $\geq (k - i - 1^2 - (k - i - 2)$ edges. Otherwise the minimum degree vertex in $G_i$ has degree $\leq k - i - 1$ and by deleting it and the other endpoint of the edge we add to the matching we remove at most $2k - 2i - 2$ edges. We therefore have
        \begin{align*}
            |E_{i + 1}| &\geq |E_i| - 2k + 2i + 2 \geq (k - i)^2 - k + i + 1 -2k +2i + 2 \\
            &= k^2 - 2k(i-1) + (i+1)^2 - k + i + 2 = (k - (i + 1))^2 -k + (i + 1) + 1
        \end{align*}
        Therefore $G_i$ always contains a perfect matching and so in every iteration $1,...,k$ an edge is added to $M$, giving us a perfect matching.
    \end{proof}

    As we are guaranteed to find a perfect matching, reconstruction is possible.
\end{proof}

\subsection{Proof of Theorem \ref{theorem:two_layer_sufficient_matching}} \label{appendix:two_layer_sufficient_matching}
\begin{proof}
    We begin by noting that this criterion can only hold if $|L|, |R| \leq k - 1$.

    The key idea of our proof is that we can construct $G_{CR}$ from $\tlc{C}$ by repeatedly contracting two vertices that are on the same side and are assigned to the same group by $p'$ or $p''$, then adding vertices to both sides of the bipartition until each side has $k$ vertices. We call these newly added vertices jokers, as in the bipartite complement they can be matched to any vertex on the other side. Let $L_c, R_c$ be $L$ and $R$ after these contractions have taken place and let $i$ be the number of contractions.

    When contracting two vertices $v, v'$ it may happen that one or both of them are endpoints of edges in $M$. In this case we remove the edges with one endpoint $v$ or $v'$ from $M$. We note that in each contraction we can loose at most $2$ edges in $M$. Let $M_c$ denote the matching after the contractions.

    We now distinguish two cases based on the number of contractions necessary to obtain $G_{CR}$ from $\tlc{C}$:

    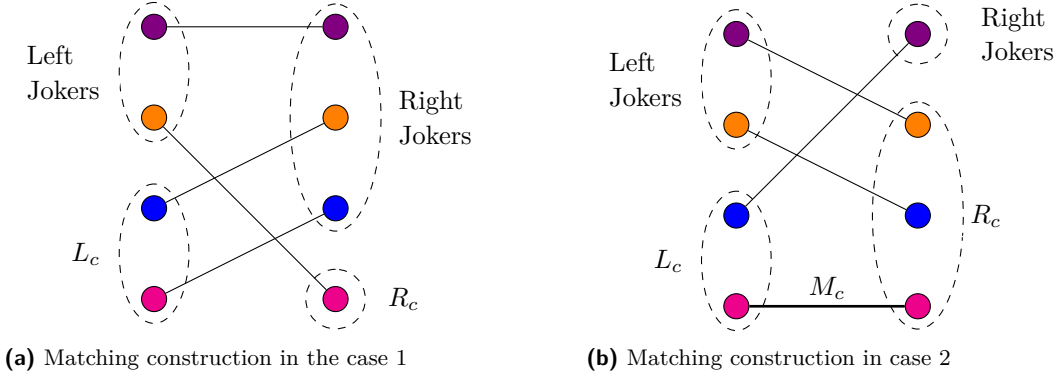
\begin{figure}[h!]
    \centering
    \subcaptionbox{Matching construction in the case $1$\label{figure:reconstruction_case_1}}[0.45\linewidth]{
        \centering
        \begin{tikzpicture}[nodes={circle,draw}, scale=0.6]
            \node[fill=magenta] (l1) at (0, 0) {};
            \node[fill=blue] (l2) at (0, 2) {};
            \node[fill=orange] (l3) at (0, 4) {};
            \node[fill=violet] (l4) at (0, 6) {};
            \node[fill=magenta] (r1) at (4, 0) {};
            \node[fill=blue] (r2) at (4, 2) {};
            \node[fill=orange] (r3) at (4, 4) {};
            \node[fill=violet] (r4) at (4, 6) {};
            
            \draw[-] (l1) -- (r2);
            \draw[-] (l2) -- (r3);
            \draw[-] (l3) -- (r1);
            \draw[-] (l4) -- (r4);

            \node[rotate=90][dashed, inner sep=3pt, circle,yscale=.5, fit={(l1) (l2)}] {};
            \node[rotate=90][dashed, inner sep=3pt, circle,yscale=.5, fit={(l3) (l4)}] {};
            \node[rotate=90][dashed, inner sep=3pt, circle,yscale=1, fit={(r1)}] {};
            \node[rotate=90][dashed, inner sep=3pt, circle,yscale=.4, fit={(r2) (r3) (r4)}] {};

            \node[color=black!0, text=black!100] at (-1.5, 1) {$L_c$};
            \node[color=black!0, text=black!100, align=left] at (-2, 5) {Left\\Jokers};

            \node[color=black!0, text=black!100] at (5.5, 0) {$R_c$};
            \node[color=black!0, text=black!100, align=left] at (6.2, 4) {Right\\Jokers};
    \end{tikzpicture}
    }
    \hfill
    \subcaptionbox{Matching construction in case $2$\label{figure:reconstruction_case_2}}[0.45\linewidth]{
        \centering
        \begin{tikzpicture}[nodes={circle,draw}, scale=0.6]
            \node[fill=magenta] (l1) at (0, 0) {};
            \node[fill=blue] (l2) at (0, 2) {};
            \node[fill=orange] (l3) at (0, 4) {};
            \node[fill=violet] (l4) at (0, 6) {};
            \node[fill=magenta] (r1) at (4, 0) {};
            \node[fill=blue] (r2) at (4, 2) {};
            \node[fill=orange] (r3) at (4, 4) {};
            \node[fill=violet] (r4) at (4, 6) {};
            
            \draw[-, line width=1pt] (l1) -- (r1);
            \draw[-] (l2) -- (r4);
            \draw[-] (l3) -- (r2);
            \draw[-] (l4) -- (r3);
            \node[color=black!0, draw opacity=0, text=black!100] at (2, 0.4) {$M_c$};

            \node[rotate=90][dashed, inner sep=3pt, circle,yscale=.5, fit={(l1) (l2)}] {};
            \node[rotate=90][dashed, inner sep=3pt, circle,yscale=.5, fit={(l3) (l4)}] {};
            \node[rotate=90][dashed, inner sep=3pt, circle,yscale=.4, fit={(r1) (r2) (r3)}] {};
            \node[rotate=90][dashed, inner sep=3pt, circle,yscale=1, fit={(r4)}] {};

            \node[color=black!0, text=black!100] at (-1.5, 1) {$L_c$};
            \node[color=black!0, text=black!100, align=left] at (-2, 5) {Left\\Jokers};

            \node[color=black!0, text=black!100] at (5.5, 2) {$R_c$};
            \node[color=black!0, text=black!100, align=left] at (6.2, 6) {Right\\Jokers};
    \end{tikzpicture}
    }
    \caption{
        The choice of perfect matching on the bipartite complement of a colour relation graph for $k = 4$
    }
    \label{fig:separator_reconstruction}
\end{figure}

    \textbf{Case 1:} $i \geq |V_C| - k$. In this case $M_c$ may be empty. However, we have $|L_c| + |R_c| \leq k$ and therefore the right side of $G_{CR}$ contains at least $|L_c|$ jokers and the left side contains at least $|R_c|$ jokers. We can then find a perfect matching by matching all vertices in $L_c$ with jokers on the right side, matching all vertices in $R_c$ with jokers on the left side and matching any remaining vertices. A visualization can be seen in Figure \ref{figure:reconstruction_case_1}.

    \textbf{Case 2:} $i < |V_C| - k$. In this case $M_c$ still contains $\geq 2 \cdot (|V_H| - k) - 1 - 2 \cdot i = 2 \cdot (|L_c| + |R_c| - k) - 1$ edges. We construct our perfect matching by taking all edges in $M_c$, matching all unmatched vertices in $L_c$ to jokers on the right side and vice versa and then match any remaining jokers. A visualization can be seen in Figure \ref{figure:reconstruction_case_2}. \\
    For this to work we must show that there are enough jokers on the right side to match to the remaining vertices in $L_c$ (the proof for $R_c$ is symmetric). For this we show that $L_c$ contains fewer vertices than the matching has edges + the number of jokers on the right.
    \begin{align*}
         && 0 &\leq |L_c| + |R_c| - k - 1 \\
        \Leftrightarrow && |L_c| &\leq 2 \cdot (|L_c| + |R_c| - k) - 1 + k - |R_c| \\
        \Rightarrow && |L_c| &\leq |M_c| + k - |R_c|
    \end{align*}
    As $i < |V_C| - k$ we must have $|L_c| + |R_c| \geq k + 1$, so the first inequality holds, implying the last.
\end{proof}

\subsection{Further Notes on \tlctitle} \label{appendix:further_applicability_two_layer}
A natural question is if we can also split \tlcname that contain negative weight edges. Unfortunately that is not the case. Given a cut set $C \subseteq E$ and $\{u, v\} = e \in C$ with $w(e) < 0$, we distinguish two cases:

\textbf{Case 1:} $|C| = 1$. In this case either the connected component containing $e$ contains only two vertices, or at least one endpoint of $e$ is a separating vertex and we can simply apply the rule for splitting biconnected components. Therefore, there is no need to consider the algorithms we developed around \tlcname.

\textbf{Case 2:} $|C| \geq 2$. In this case let $\{u', v'\} = e' \in C, e \not = e'$ and we assume that w.l.o.g. $u \not = u'$, $u, u'$ are on one side of the cut set and $v, v'$ are on the other. If $w(e') < 0$, then if our solutions for $G_1, G_2$ require that $p(u) \not = p(u')$ and $p(v) = p(v')$, we must cut one of $e$ and $e'$. If instead the solutions require $p(u) = p(u')$ and $p(v) = p(v')$ it is possible to cut neither.

As such we would somehow have to encode the trade off between cutting neither of $e, e'$ and performing better on $G_1$ and $G_2$, which makes the techniques we developed unusable. If $w(e') > 0$ a similar argument can be made.
\section{Computational Experiments}

\subsection{Detailed Overview of our Preprocessing Order} \label{appendix:preprocessing_order}
We chose the order, in which the rules are applied, based on the following criteria:
\begin{itemize}
    \item Does a rule require and/or maintain local positivity and/or unit weight?
    \item Can applicability of the rule be checked locally around a single vertex or do we have to observe the whole graph?
    \item How fast is it to check and apply a rule?
\end{itemize}
Our core idea is that all rules that require local unit weight and positivity should be applied to exhaustion before any rules that may break either property. We then run rules that may break local unit weight, but not local positivity and finally we run all remaining rules. This is to ensure that no rule causes problems with the prerequisites of another rule.

Within each group of rules, we first applied those that can be checked locally. For each such rule, we maintain a queue of candidate vertices where the rule may be applicable. If, for a vertex $v$, an incident edge changes, we add $v$ back to the queues for locally checkable rules again. For rules that run on the whole graph, once applied once, we mark them as blocked and only unblock them once all other rules have run to exhaustion. The only exception to this is the separator rules, which we unblock before moving from unit weight relying rules to rules that break unit weight and again when moving from rules that require local positivity to rules that break local positivity.

After some experimenting, we arrived at this order for the data separation/reduction rules:

\begin{enumerate}
    \item Splitting into connected and biconnected components
    \item Removal of locally positive low degree nodes
    \item Splitting at \tlcname (see Section \ref{section:two_layer_separators})
    \item When reaching this point for the first time, reactivate all data separation rules and start from the top
    \item Removing unit weight cliques (see Theorem \ref{theorem:max_cut_cliques})
    \item Rules requiring and maintaining local positivity
    \begin{enumerate}
        \item Contracting negative dominating edges (see Theorem \ref{theorem:dominating_mc})
        \item Contracting negative triangles (see Theorem \ref{theorem:dominating_triangle_mkc})
    \end{enumerate}
    \item When reaching this point for the first time, reactivate all data separation rules and start from the top
    \item Remaining rules
    \begin{enumerate}
        \item Removing small leafs of SPQR-trees (see Theorem \ref{theorem:two_vertex_cut})
        \item Split \tlcshort s where one side is small and can be solved during preprocessing.
    \end{enumerate}
    \item When reaching this point for the first time, reactivate all rules and start from the top
\end{enumerate}
Any time a rule shrinks or splits the graph, we start again from the top.

\subsection{Instance Selection} \label{appendix:instance_selection}
See Table \ref{tab:instance_overview} for on overview over all instances.

For the frequency assignment instances we removed the frequency constraints on the nodes and relaxed all edge constraints to ``the endpoints must have different colours''. If the instances came with costs for breaking an edge constraint, we used that as the edge weight. Otherwise we set all edge weights to $0$. The instances that come with their own costs assign one of $4$ levels of importance to the edges. These correspond with costs $1, 10, 100, 10.00$, $1, 100, 10.000, 1.000.000$ or $1, 1.000, 1.000.000, 1.000.000.000$, depending on the instance.

\begin{table}[h!]
    \centering
    \begin{tabular}{llrrrrrr}
        \toprule
        Group & Instance & $n$ & $m$ & $\underline{d}$ & $\overline{d}$ & $\underline{w}$ & $\overline{w}$ \\
        \midrule
        \multirow{ 13 }{*}{ easy }         & soc-firm & 33 & 91 & 1 & 16 & 1 & 2 \\
         & g001207 & 84 & 149 & 1 & 5 & 1 & 100.000 \\
         & g000981 & 110 & 188 & 2 & 6 & 1 & 100.000 \\
         & ENZYMES295 & 123 & 139 & 1 & 5 & 1 & 1 \\
         & g000292 & 212 & 381 & 2 & 4 & 5 & 13 \\
         & g000302 & 317 & 476 & 1 & 4 & 5 & 13 \\
         & rt-twitter-copen & 761 & 1,029 & 1 & 37 & 1 & 1 \\
         & g001918 & 777 & 1,239 & 1 & 4 & 5 & 13 \\
         & imgseg\_271031 & 900 & 1,027 & 1 & 518 & 93 & 285,968 \\
         & imgseg\_106025 & 1,565 & 2,629 & 1 & 902 & 93 & 136,834 \\
         & g000677 & 17,127 & 27,352 & 1 & 4 & 1 & 126 \\
         & g001075 & 27,019 & 39,407 & 1 & 4 & 1 & 228.668 \\
         & g000087 & 38,418 & 71,657 & 2 & 4 & 1 & 198 \\
        \midrule
        \multirow{ 15 }{*}{ medium }         & ca-netscience & 379 & 914 & 1 & 34 & 1 & 1 \\
         & bio-celegans & 453 & 2,025 & 1 & 237 & 1 & 1 \\
         & bio-diseasome & 516 & 1,188 & 1 & 50 & 1 & 1 \\
         & bio-DM-LC & 658 & 1,129 & 1 & 50 & 1 & 1 \\
         & road-euroroad & 1,174 & 1,417 & 1 & 10 & 1 & 1 \\
         & imgseg\_35058 & 1,274 & 1,806 & 1 & 587 & -55,510 & 112,271 \\
         & bio-yeast & 1,458 & 1,948 & 1 & 56 & 1 & 1 \\
         & ca-CSphd & 1,882 & 1,740 & 1 & 46 & 1 & 1 \\
         & ego-facebook & 2,888 & 2,981 & 1 & 769 & 1 & 1 \\
         & imgseg\_105019 & 3,548 & 4,325 & 1 & 2,753 & 109 & 236,593 \\
         & inf-power & 4,941 & 6,594 & 1 & 19 & 1 & 1 \\
         & ca-Erdos992 & 5,094 & 7,515 & 1 & 61 & 1 & 1 \\
         & imgseg\_374020 & 5,735 & 8,722 & 1 & 2,213 & -46,639 & 407,957 \\
         & imgseg\_147062 & 28,552 & 65,453 & 1 & 925 & -1,567 & 67,209 \\
         & road-luxembourg-osm & 114,599 & 119,666 & 1 & 6 & 1 & 1 \\
        \midrule
        \multirow{ 6 }{*}{ hard }         & web-Stanford & 281,903 & 1,992,636 & 1 & 38,625 & 1 & 2 \\
         & web-it-2004 & 509,338 & 7,178,413 & 1 & 469 & 1 & 1 \\
         & ca-coauthors-dblp & 540,486 & 15,245,729 & 1 & 3,299 & 1 & 1 \\
         & web-google & 870,204 & 4,258,481 & 1 & 6,332 & 1 & 2 \\
         & ca-IMDB & 896,305 & 3,782,447 & 1 & 1,590 & 1 & 2 \\
         & inf-road-central & 14,081,816 & 16,933,413 & 1 & 8 & 1 & 1 \\
    \bottomrule
    \end{tabular}
    \caption{Overview of Instances used. $\underline{d}, \overline{d}, \underline{w}$ and $\overline{w}$ are the minimum and maximum degree and weight respectively.  If an instance contained parallel edges, we replaced them by a single edge with the sum of their edge weights as the weight. This leads to some of our graphs having lower edge counts and higher maximum weights than reported by their original source. The table continues in table \ref{tab:instance_overview_2}}
    \label{tab:instance_overview}
\end{table}

\begin{table}[]
    \centering
\begin{tabular}{llrrrrrr}
        \toprule
        Group & Instance & $n$ & $m$ & $\underline{d}$ & $\overline{d}$ & $\underline{w}$ & $\overline{w}$ \\
        \midrule
        \multirow{ 18 }{*}{ torus }         & t2g10\_5555 & 100 & 200 & 4 & 4 & -294.541 & 290.339 \\
         & t2g10\_6666 & 100 & 200 & 4 & 4 & -239.344 & 238.268 \\
         & t2g10\_7777 & 100 & 200 & 4 & 4 & -238.936 & 301.004 \\
         & t2g15\_5555 & 225 & 450 & 4 & 4 & -294.541 & 290.339 \\
         & t2g15\_6666 & 225 & 450 & 4 & 4 & -240.195 & 268.055 \\
         & t2g15\_7777 & 225 & 450 & 4 & 4 & -247.819 & 375.001 \\
         & t2g20\_5555 & 400 & 800 & 4 & 4 & -294.541 & 308.059 \\
         & t2g20\_6666 & 400 & 800 & 4 & 4 & -271.149 & 315.291 \\
         & t2g20\_7777 & 400 & 800 & 4 & 4 & -288.410 & 375.001 \\
         & t3g5\_5555 & 125 & 375 & 6 & 6 & -294.541 & 290.339 \\
         & t3g5\_6666 & 125 & 375 & 6 & 6 & -240.195 & 268.055 \\
         & t3g5\_7777 & 125 & 375 & 6 & 6 & -238.936 & 375.001 \\
         & t3g6\_5555 & 216 & 648 & 6 & 6 & -294.541 & 308.059 \\
         & t3g6\_6666 & 216 & 648 & 6 & 6 & -265.601 & 271.240 \\
         & t3g6\_7777 & 216 & 648 & 6 & 6 & -288.410 & 375.001 \\
         & t3g7\_5555 & 343 & 1,029 & 6 & 6 & -294.541 & 308.059 \\
         & t3g7\_6666 & 343 & 1,029 & 6 & 6 & -271.149 & 315.291 \\
         & t3g7\_7777 & 343 & 1,029 & 6 & 6 & -298.103 & 375.001 \\
        \midrule
        \multirow{ 32 }{*}{ fap }         & DUTtest1\_200 & 200 & 1,171 & 6 & 21 & 1 & 1 \\
         & DUTtest1\_200 & 200 & 1,143 & 5 & 24 & 1 & 1 \\
         & DUTtest1\_200 & 200 & 1,160 & 7 & 19 & 1 & 1 \\
         & DUTtest1\_200 & 200 & 1,142 & 5 & 17 & 1 & 1 \\
         & DUTtest1\_200 & 200 & 1,125 & 2 & 23 & 1 & 1 \\
         & DUTtest1\_916 & 916 & 5,177 & 1 & 39 & 1 & 1 \\
         & DUTtest1\_916 & 916 & 5,173 & 1 & 24 & 1 & 1 \\
         & DUTtest1\_916 & 916 & 5,262 & 2 & 23 & 1 & 1 \\
         & DUTtest1\_916 & 916 & 5,183 & 2 & 24 & 1 & 1 \\
         & DUTtest1\_916 & 916 & 5,213 & 1 & 24 & 1 & 1 \\
         & SURPRISE\_01 & 200 & 1,134 & 1 & 22 & 1 & 1 \\
         & SURPRISE\_02 & 400 & 2,245 & 3 & 31 & 1 & 1 \\
         & SURPRISE\_03 & 200 & 1,134 & 6 & 18 & 1 & 1 \\
         & SURPRISE\_04 & 400 & 2,244 & 4 & 20 & 1 & 1 \\
         & SURPRISE\_05 & 200 & 1,134 & 1 & 22 & 1 & 1 \\
         & SURPRISE\_06 & 400 & 2,170 & 1 & 24 & 1 & 1 \\
         & SURPRISE\_07 & 400 & 2,170 & 1 & 24 & 1 & 1 \\
         & SURPRISE\_08 & 680 & 3,757 & 1 & 22 & 1 & 1 \\
         & SURPRISE\_09 & 916 & 5,246 & 1 & 37 & 1 & 1 \\
         & SURPRISE\_10 & 680 & 3,907 & 3 & 24 & 1 & 1 \\
         & SURPRISE\_11 & 680 & 3,757 & 1 & 22 & 1 & 1 \\
         & CELAR\_01 & 916 & 5,548 & 1 & 61 & 1 & 1 \\
         & CELAR\_02 & 200 & 1,235 & 1 & 44 & 1 & 1 \\
         & CELAR\_03 & 400 & 2,760 & 3 & 61 & 1 & 1 \\
         & CELAR\_04 & 680 & 3,967 & 1 & 62 & 100.000 & 100.000 \\
         & CELAR\_05 & 400 & 2,598 & 1 & 59 & 100.000 & 100.000 \\
         & CELAR\_06 & 200 & 1,322 & 1 & 44 & 1 & 100.000 \\
         & CELAR\_07 & 400 & 2,865 & 3 & 62 & 1 & 100.000.000 \\
         & CELAR\_08 & 916 & 5,744 & 1 & 62 & 1 & 400 \\
         & CELAR\_09 & 680 & 4,103 & 1 & 62 & 1 & 100.000 \\
         & CELAR\_10 & 680 & 4,103 & 1 & 62 & 1 & 100.000 \\
         & CELAR\_11 & 680 & 4,103 & 1 & 62 & 100.000 & 100.000 \\
    \bottomrule
    \end{tabular}
    \caption{Continuation of table \ref{tab:instance_overview}. Overview of Instances used, $\underline{d}, \overline{d}, \underline{w}$ and $\overline{w}$ are the minimum and maximum degree and weight respectively. If an instance contained parallel edges, we replaced them by a single edge with the sum of their edge weights as the weight. This leads to some of our graphs having lower edge counts and higher maximum weights than reported by their original source.}
    \label{tab:instance_overview_2}
\end{table}

\subsection{Results on the Easy Instance Set} \label{appendix:easy}
See Table \ref{tab:benchmark_easy}.

\begin{table}[h!]
\begin{scriptsize}
\begin{center}
    {\tabcolsep2.5pt
\begin{tabular*}{\textwidth}{l@{\extracolsep{\fill}}rr@{\extracolsep{\fill}}rrrr@{\extracolsep{\fill}}rrrr}
\toprule
 & \multicolumn{2}{c}{\textbf{Gurobi}} & \multicolumn{4}{c}{\textbf{Naive}} & \multicolumn{4}{c}{\textbf{Our}} \\
\cmidrule(){2-3}\cmidrule(){4-7}\cmidrule(){8-11}
$k$ & wins & t[s] & $|V|[\%]$ & $|E|[\%]$ & wins & t[s] & $|V|[\%]$ & $|E|[\%]$ & wins & t[s] \\
\midrule
3 & 0 & 195 & 15.27 & 18.50 & \textbf{10} & 123 & \textbf{14.24} & \textbf{16.98} & 9 & \textbf{64} \\
4 & 0 & 77 & \textbf{4.69} & \textbf{5.77} & \textbf{11} & \textbf{0} & \textbf{4.69} & \textbf{5.77} & 10 & \textbf{0} \\
5 & 0 & 75 & \textbf{1.40} & \textbf{1.27} & 10 & \textbf{0} & \textbf{1.40} & \textbf{1.27} & \textbf{12} & \textbf{0} \\
6 & 0 & 79 & \textbf{0.00} & \textbf{0.00} & \textbf{13} & \textbf{0} & \textbf{0.00} & \textbf{0.00} & 11 & \textbf{0} \\
7 & 0 & 84 & \textbf{0.00} & \textbf{0.00} & \textbf{13} & \textbf{0} & \textbf{0.00} & \textbf{0.00} & 10 & \textbf{0} \\
8 & 0 & 87 & \textbf{0.00} & \textbf{0.00} & \textbf{13} & \textbf{0} & \textbf{0.00} & \textbf{0.00} & 11 & \textbf{0} \\
10 & 0 & 94 & \textbf{0.00} & \textbf{0.00} & \textbf{12} & \textbf{0} & \textbf{0.00} & \textbf{0.00} & 11 & \textbf{0} \\
12 & 0 & 97 & \textbf{0.00} & \textbf{0.00} & \textbf{12} & \textbf{0} & \textbf{0.00} & \textbf{0.00} & \textbf{12} & \textbf{0} \\
\bottomrule
\end{tabular*}

}
\end{center}
\end{scriptsize}
    \caption{Number of wins, average runtime and average percentage of vertices and edges remaining after \textbf{naive} and \textbf{our} preprocessing algorithm on the \textbf{easy} data set.}
    \label{tab:benchmark_easy}
\end{table}

\subsection{Heuristic Experiments} \label{appendix:big_full}
See Table \ref{tab:benchmark_big_full}.

\begin{table}[h!]
\begin{scriptsize}
\begin{center}
    {\tabcolsep2.5pt
\begin{tabular*}{\textwidth}{ll@{\extracolsep{\fill}}r@{\extracolsep{\fill}}rrrr@{\extracolsep{\fill}}rrrr}
\toprule
 & & \textbf{Heuristic} & \multicolumn{4}{c}{\textbf{Naive}} & \multicolumn{4}{c}{\textbf{Our}} \\
\cmidrule(){4-7}\cmidrule(){8-11}
$k$ & \textbf{instance} & best value & $|V|[\%]$ & $|E|[\%]$ & bvi & pr[s] & $|V|[\%]$ & $|E|[\%]$ & bvi & pr[s] \\
\midrule 
 \multirow{ 6 }{*}{ 3 } & ca-IMDB & 3.618.918 & \textbf{46.43} & \textbf{84.61} & 56.689 & 3 & \textbf{46.43} & \textbf{84.61} & \textbf{58.046} & 65 \\
 & ca-coauthors-dblp & 10.812.538 & 96.45 & 99.79 & 4.386 & 15 & \textbf{82.95} & \textbf{95.83} & \textbf{6.554} & 839 \\
 & inf-road-central & 16.698.657 & 0.02 & 0.02 & \textbf{234.703} & 19 & \textbf{0.01} & \textbf{0.02} & \textbf{234.703} & 19 \\
 & web-Stanford & 2.121.484 & 77.01 & 94.97 & 10.234 & 2 & \textbf{73.96} & \textbf{93.38} & \textbf{10.642} & 175 \\
 & web-google & 4.382.696 & 66.56 & 90.52 & 31.780 & 8 & \textbf{62.28} & \textbf{87.14} & \textbf{34.608} & 1.176 \\
 & web-it-2004 & 5.171.402 & 92.00 & 99.05 & \textbf{6.583} & 4 & \textbf{11.26} & \textbf{35.92} & \textbf{6.583} & 43 \\
\midrule 
 \multirow{ 6 }{*}{ 4 } & ca-IMDB & 3.687.347 & \textbf{39.58} & \textbf{79.82} & 50.525 & 3 & \textbf{39.58} & \textbf{79.82} & \textbf{50.589} & 63 \\
 & ca-coauthors-dblp & 12.070.341 & 93.80 & 99.58 & \textbf{3.618} & 12 & \textbf{84.78} & \textbf{97.13} & 1.480 & 1.008 \\
 & inf-road-central & 16.930.983 & \textbf{0.00} & \textbf{0.00} & \textbf{2.430} & 17 & \textbf{0.00} & \textbf{0.00} & \textbf{2.430} & 17 \\
 & web-Stanford & 2.213.473 & 65.49 & 90.45 & \textbf{4.314} & 2 & \textbf{63.48} & \textbf{89.19} & 3.096 & 178 \\
 & web-google & 4.679.030 & 56.32 & 84.80 & \textbf{22.832} & 8 & \textbf{53.47} & \textbf{82.05} & 15.212 & 1.447 \\
 & web-it-2004 & 5.727.019 & 91.44 & 98.96 & \textbf{5.208} & 4 & \textbf{11.04} & \textbf{35.92} & \textbf{5.208} & 44 \\
\midrule 
 \multirow{ 6 }{*}{ 5 } & ca-IMDB & 3.720.314 & \textbf{34.44} & \textbf{75.04} & \textbf{38.148} & 3 & \textbf{34.44} & \textbf{75.04} & 37.480 & 58 \\
 & ca-coauthors-dblp & 12.802.326 & 91.25 & 99.30 & \textbf{4.130} & 12 & \textbf{84.43} & \textbf{97.32} & 2.916 & 1.696 \\
 & inf-road-central & 16.933.413 & \textbf{0.00} & \textbf{0.00} & \textbf{0} & 16 & \textbf{0.00} & \textbf{0.00} & \textbf{0} & 16 \\
 & web-Stanford & 2.250.780 & 57.19 & 86.13 & 2.702 & 2 & \textbf{55.78} & \textbf{85.15} & \textbf{2.869} & 178 \\
 & web-google & 4.826.044 & 48.47 & 78.98 & \textbf{17.623} & 7 & \textbf{46.26} & \textbf{76.59} & 7.947 & 1.269 \\
 & web-it-2004 & 6.058.523 & 90.96 & 98.87 & 4.370 & 4 & \textbf{10.88} & \textbf{35.88} & \textbf{4.371} & 47 \\
\midrule 
 \multirow{ 6 }{*}{ 6 } & ca-IMDB & 3.743.987 & \textbf{30.07} & \textbf{69.98} & 24.978 & 3 & \textbf{30.07} & \textbf{69.98} & \textbf{25.065} & 55 \\
 & ca-coauthors-dblp & 13.278.589 & 88.81 & 98.98 & \textbf{4.449} & 11 & \textbf{83.35} & \textbf{97.32} & 1.241 & 1.772 \\
 & inf-road-central & 16.933.413 & \textbf{0.00} & \textbf{0.00} & \textbf{0} & 16 & \textbf{0.00} & \textbf{0.00} & \textbf{0} & 16 \\
 & web-Stanford & 2.270.156 & 40.99 & 75.84 & 1.941 & 1 & \textbf{40.11} & \textbf{75.14} & \textbf{1.993} & 140 \\
 & web-google & 4.914.053 & 41.97 & 73.02 & 11.713 & 7 & \textbf{40.24} & \textbf{70.97} & \textbf{12.162} & 1.223 \\
 & web-it-2004 & 6.279.822 & 90.36 & 98.74 & 3.884 & 4 & \textbf{10.71} & \textbf{35.84} & \textbf{3.886} & 43 \\
\midrule 
 \multirow{ 6 }{*}{ 7 } & ca-IMDB & 3.760.031 & \textbf{26.13} & \textbf{64.50} & 15.666 & 3 & \textbf{26.13} & \textbf{64.50} & \textbf{15.719} & 49 \\
 & ca-coauthors-dblp & 13.612.054 & 86.50 & 98.60 & \textbf{4.809} & 11 & \textbf{82.01} & \textbf{97.24} & 773 & 2.286 \\
 & inf-road-central & 16.933.413 & \textbf{0.00} & \textbf{0.00} & \textbf{0} & 17 & \textbf{0.00} & \textbf{0.00} & \textbf{0} & 17 \\
 & web-Stanford & 2.282.032 & 33.52 & 69.91 & 1.436 & 1 & \textbf{33.08} & \textbf{69.48} & \textbf{1.468} & 118 \\
 & web-google & 4.969.060 & 36.46 & 67.09 & 8.982 & 6 & \textbf{34.98} & \textbf{65.17} & \textbf{9.412} & 1.128 \\
 & web-it-2004 & 6.434.136 & 90.01 & 98.66 & 3.165 & 4 & \textbf{10.65} & \textbf{35.82} & \textbf{3.166} & 41 \\
\midrule 
 \multirow{ 6 }{*}{ 8 } & ca-IMDB & 3.769.745 & \textbf{22.62} & \textbf{58.81} & 9.385 & 3 & \textbf{22.62} & \textbf{58.81} & \textbf{9.437} & 44 \\
 & ca-coauthors-dblp & 13.857.795 & 84.22 & 98.17 & \textbf{3.982} & 11 & \textbf{80.41} & \textbf{96.93} & 131 & 1.855 \\
 & inf-road-central & 16.933.413 & \textbf{0.00} & \textbf{0.00} & \textbf{0} & 17 & \textbf{0.00} & \textbf{0.00} & \textbf{0} & 16 \\
 & web-Stanford & 2.289.432 & 29.56 & 66.31 & 1.152 & 2 & \textbf{29.21} & \textbf{65.93} & \textbf{1.192} & 100 \\
 & web-google & 5.007.504 & 31.35 & 60.78 & 5.939 & 6 & \textbf{30.12} & \textbf{59.05} & \textbf{6.241} & 970 \\
 & web-it-2004 & 6.554.618 & 89.39 & 98.48 & 2.809 & 4 & \textbf{10.54} & \textbf{35.77} & \textbf{2.812} & 42 \\
\midrule 
 \multirow{ 6 }{*}{ 10 } & ca-IMDB & 3.777.961 & \textbf{16.29} & \textbf{46.38} & \textbf{3.865} & 2 & \textbf{16.29} & \textbf{46.38} & 3.858 & 37 \\
 & ca-coauthors-dblp & 14.192.268 & 79.89 & 97.19 & \textbf{4.717} & 10 & \textbf{77.07} & \textbf{96.21} & 2.338 & 1.793 \\
 & inf-road-central & 16.933.413 & \textbf{0.00} & \textbf{0.00} & \textbf{0} & 16 & \textbf{0.00} & \textbf{0.00} & \textbf{0} & 16 \\
 & web-Stanford & 2.298.360 & 23.20 & 59.40 & \textbf{912} & 2 & \textbf{22.99} & \textbf{59.13} & 907 & 63 \\
 & web-google & 5.051.920 & 22.09 & 47.17 & 4.502 & 4 & \textbf{21.26} & \textbf{45.79} & \textbf{4.768} & 561 \\
 & web-it-2004 & 6.711.894 & 88.01 & 98.02 & 2.149 & 4 & \textbf{10.26} & \textbf{35.64} & \textbf{2.150} & 42 \\
\midrule 
 \multirow{ 6 }{*}{ 12 } & ca-IMDB & 3.780.756 & \textbf{9.94} & \textbf{31.03} & \textbf{1.636} & 2 & \textbf{9.94} & \textbf{31.03} & \textbf{1.636} & 24 \\
 & ca-coauthors-dblp & 14.409.950 & 75.93 & 96.08 & \textbf{4.927} & 10 & \textbf{73.83} & \textbf{95.30} & 3.079 & 1.371 \\
 & inf-road-central & 16.933.413 & \textbf{0.00} & \textbf{0.00} & \textbf{0} & 16 & \textbf{0.00} & \textbf{0.00} & \textbf{0} & 16 \\
 & web-Stanford & 2.303.794 & 18.70 & 53.02 & 791 & 1 & \textbf{18.54} & \textbf{52.78} & \textbf{855} & 51 \\
 & web-google & 5.075.045 & 15.00 & 34.68 & 3.711 & 4 & \textbf{14.44} & \textbf{33.59} & \textbf{3.878} & 377 \\
 & web-it-2004 & 6.812.023 & 87.38 & 97.78 & 1.649 & 4 & \textbf{10.22} & \textbf{35.62} & \textbf{1.652} & 44 \\
\bottomrule
\end{tabular*}
}
\end{center}
\end{scriptsize}
    \caption{Best value found by our local search heuristic within $2$ hours, percentage of vertices and edges remaining, improvement over the straight heuristic (bvi) and time spend preprocessing (pr) of \textbf{naive} and \textbf{our} preprocessing algorithm on the \textbf{big} data set.}
    \label{tab:benchmark_big_full}
\end{table}

\subsection{Ablation Experiments} \label{appendix:ablation_full}
See Table \ref{tab:ablation_full}.

\begin{table}[h!]
\begin{scriptsize}
\begin{center}
    {\tabcolsep2.5pt
\begin{tabular*}{\textwidth}{ll@{\extracolsep{\fill}}rr@{\extracolsep{\fill}}rr@{\extracolsep{\fill}}rr@{\extracolsep{\fill}}rr@{\extracolsep{\fill}}rr@{\extracolsep{\fill}}rr}
\toprule
 & & \textbf{All} & \textbf{Naive} & \multicolumn{2}{c}{\textbf{NoDom}} & \multicolumn{2}{c}{\textbf{NoClq}} & \multicolumn{2}{c}{\textbf{noBicon}} & \multicolumn{2}{c}{\textbf{No\tlcshort}} & \multicolumn{2}{c}{\textbf{No\tlcshort S}} \\
\cmidrule(){3-4}\cmidrule(){5-6}\cmidrule(){7-8}\cmidrule(){9-10}\cmidrule(){11-12}\cmidrule(){13-14}
$k$ & dataset & $|E|[\%]$ & $|E|[\%]$ & $|E|[\%]$ & wins & $|E|[\%]$ & wins & $|E|[\%]$ & wins & $|E|[\%]$ & wins & $|E|[\%]$ & wins  \\
\midrule 
 \multirow{ 4 }{*}{ 3 } & easy & 16.98 & 18.50 & 16.98 & \textbf{6} & 17.03 & \textbf{6} & \textbf{18.37} & 8 & 16.98 & 7 & 16.98 & 8 \\
 & medium & 22.60 & 36.74 & 24.59 & 3 & \textbf{27.26} & \textbf{1} & 26.86 & 3 & 22.53 & 6 & 22.67 & 6 \\
 & torus & 95.30 & 100.00 & \textbf{100.00} & 3 & 95.30 & \textbf{1} & 95.30 & 4 & 95.30 & 5 & 95.30 & 5 \\
 & fap & 96.70 & 99.74 & 96.72 & \textbf{11} & 97.55 & \textbf{11} & \textbf{98.08} & 12 & 96.70 & 17 & 96.70 & 14 \\
\midrule 
 \multirow{ 4 }{*}{ 4 } & easy & 5.77 & 5.77 & \textbf{5.77} & \textbf{7} & \textbf{5.77} & 10 & \textbf{5.77} & 9 & \textbf{5.77} & 9 & \textbf{5.77} & 10 \\
 & medium & 12.35 & 21.73 & 14.28 & 5 & \textbf{16.52} & \textbf{4} & 13.91 & 7 & 12.32 & 5 & 12.35 & 6 \\
 & torus & 95.26 & 100.00 & \textbf{100.00} & 6 & 95.26 & 3 & 95.26 & 4 & 95.26 & \textbf{2} & 95.26 & 3 \\
 & fap & 95.73 & 98.05 & 95.73 & 11 & 96.22 & 10 & \textbf{96.60} & \textbf{6} & 95.09 & 7 & 95.82 & 14 \\
\midrule 
 \multirow{ 4 }{*}{ 5 } & easy & 1.27 & 1.27 & \textbf{1.27} & 12 & \textbf{1.27} & 11 & \textbf{1.27} & 11 & \textbf{1.27} & \textbf{10} & \textbf{1.27} & 11 \\
 & medium & 7.31 & 15.36 & 9.42 & \textbf{4} & 8.48 & 5 & \textbf{9.78} & 7 & 7.34 & 6 & 7.31 & 8 \\
 & torus & 86.60 & 94.75 & \textbf{93.83} & 7 & 86.60 & 2 & 86.80 & \textbf{1} & 86.67 & \textbf{1} & 87.10 & 7 \\
 & fap & 89.34 & 93.02 & 89.34 & 12 & 90.25 & 11 & \textbf{90.69} & \textbf{8} & 89.71 & 11 & 90.02 & 16 \\
\midrule 
 \multirow{ 4 }{*}{ 6 } & easy & 0.00 & 0.00 & \textbf{0.00} & 13 & \textbf{0.00} & \textbf{11} & \textbf{0.00} & \textbf{11} & \textbf{0.00} & \textbf{11} & \textbf{0.00} & \textbf{11} \\
 & medium & 6.05 & 11.53 & \textbf{8.40} & 8 & 7.88 & 7 & 6.56 & \textbf{6} & 6.06 & 8 & 6.05 & \textbf{6} \\
 & torus & 83.52 & 94.75 & \textbf{91.85} & 8 & 83.52 & \textbf{1} & 83.94 & 3 & 82.98 & 4 & 85.06 & 2 \\
 & fap & 68.66 & 76.38 & 68.58 & 5 & 70.49 & 6 & 69.08 & 7 & 68.86 & \textbf{4} & \textbf{73.22} & 14 \\
\midrule 
 \multirow{ 4 }{*}{ 7 } & easy & 0.00 & 0.00 & \textbf{0.00} & \textbf{11} & \textbf{0.00} & 12 & \textbf{0.00} & \textbf{11} & \textbf{0.00} & 12 & \textbf{0.00} & 12 \\
 & medium & 4.20 & 9.21 & \textbf{6.60} & 10 & 5.48 & 7 & 4.71 & 11 & 4.20 & \textbf{6} & 4.20 & 9 \\
 & torus & 77.67 & 92.75 & \textbf{87.10} & 8 & 77.67 & 2 & 78.42 & 4 & 77.35 & 3 & 79.95 & \textbf{1} \\
 & fap & 55.94 & 72.03 & 55.94 & \textbf{5} & 59.64 & 6 & 55.91 & 6 & 56.37 & 7 & \textbf{67.07} & 12 \\
\midrule 
 \multirow{ 4 }{*}{ 8 } & easy & 0.00 & 0.00 & \textbf{0.00} & \textbf{11} & \textbf{0.00} & \textbf{11} & \textbf{0.00} & 12 & \textbf{0.00} & 12 & \textbf{0.00} & \textbf{11} \\
 & medium & 3.07 & 7.34 & \textbf{5.54} & 10 & 3.82 & \textbf{8} & 3.60 & 10 & 3.09 & \textbf{8} & 3.07 & \textbf{8} \\
 & torus & 75.39 & 92.75 & \textbf{85.98} & 9 & 75.39 & 3 & 76.08 & 3 & 74.48 & 2 & 78.40 & \textbf{1} \\
 & fap & 43.75 & 62.83 & 44.13 & \textbf{3} & 48.12 & 4 & 44.05 & \textbf{3} & 39.12 & 6 & \textbf{59.04} & 16 \\
\midrule 
 \multirow{ 4 }{*}{ 10 } & easy & 0.00 & 0.00 & \textbf{0.00} & 12 & \textbf{0.00} & \textbf{11} & \textbf{0.00} & \textbf{11} & \textbf{0.00} & 12 & \textbf{0.00} & 12 \\
 & medium & 1.82 & 5.66 & \textbf{4.31} & 10 & 2.13 & \textbf{9} & 2.30 & 10 & 1.88 & \textbf{9} & 1.82 & 11 \\
 & torus & 73.17 & 92.75 & \textbf{84.18} & 9 & 73.17 & \textbf{2} & 74.12 & \textbf{2} & 72.68 & 3 & 76.63 & \textbf{2} \\
 & fap & 25.92 & 50.01 & 25.92 & 6 & 37.08 & \textbf{3} & 25.92 & 4 & 23.03 & 4 & \textbf{40.12} & 18 \\
\midrule 
 \multirow{ 4 }{*}{ 12 } & easy & 0.00 & 0.00 & \textbf{0.00} & 13 & \textbf{0.00} & \textbf{11} & \textbf{0.00} & \textbf{11} & \textbf{0.00} & \textbf{11} & \textbf{0.00} & \textbf{11} \\
 & medium & 0.27 & 3.91 & \textbf{2.87} & 11 & 0.27 & 11 & 0.75 & 10 & 0.29 & \textbf{8} & 0.27 & 12 \\
 & torus & 72.71 & 92.75 & \textbf{84.30} & 9 & 72.71 & 3 & 74.13 & \textbf{2} & 72.28 & \textbf{2} & 75.78 & \textbf{2} \\
 & fap & 17.59 & 41.25 & 17.64 & 18 & \textbf{37.59} & \textbf{5} & 17.62 & 11 & 18.87 & 7 & 20.67 & 11 \\
\bottomrule
\end{tabular*}
}
\end{center}
\end{scriptsize}
    \caption{Percentage of edges remaining and number of wins comparing full preprocessing except dominating edges (\textbf{NoDom}), cliques (\textbf{noClq}), biconnectors (\textbf{noBicon}), \tlcname (\textbf{no\tlcshort}) and \tlcname with solving the small side (\textbf{no\tlcshort S}). Naive and full preprocessing for reference, the worst values are highlighted.}
    \label{tab:ablation_full}
\end{table}

\end{document}